\documentclass[journal]{IEEEtran}

\IEEEoverridecommandlockouts

\usepackage{cite}
\usepackage{amsmath, amssymb, amsfonts, bm}
\usepackage{amsthm}
\usepackage{algorithmic}
\usepackage[ruled,lined]{algorithm2e}
\usepackage{setspace}
\usepackage{indentfirst}
\usepackage{graphicx}
\usepackage{epstopdf}
\usepackage{textcomp}
\usepackage{xcolor}
\usepackage{tikz}
\usetikzlibrary{arrows, shapes, chains}
\usepackage{subfigure}
\usepackage{mathrsfs}
\usepackage{verbatim}
\usepackage{url}
\usepackage{multirow}

\usepackage{color}
\definecolor{color1}{RGB}{128, 0, 0}
\definecolor{color2}{RGB}{0, 0, 128}
\usepackage[colorlinks, linkcolor=color1, anchorcolor=blue, citecolor=color1, urlcolor = color2]{hyperref}

\newtheorem{prop}{Proposition}

\theoremstyle{remark}

\begin{document}

\title{\fontsize{23.5pt}{12pt}\selectfont Federated Learning via Intelligent Reflecting Surface}

\author{\IEEEauthorblockN{Zhibin Wang, \IEEEmembership{Student Member, IEEE}, Jiahang Qiu, Yong Zhou, \IEEEmembership{Member, IEEE}, Yuanming Shi, \IEEEmembership{Member, IEEE}, \\
		 Liqun Fu, \IEEEmembership{Senior Member, IEEE}, Wei Chen, \IEEEmembership{Senior Member, IEEE}, and Khaled B. Letaief, \IEEEmembership{Fellow, IEEE}}
	\thanks{
		Z. Wang, Y. Zhou, and Y. Shi are with the School of Information Science and Technology, ShanghaiTech University, Shanghai, 201210, China (E-mail: \{wangzhb, zhouyong, shiym\}@shanghaitech.edu.cn).}
	\thanks{
		J. Qiu and L. Fu are with the School of Informatics, and Key Laboratory of Underwater Acoustic Communication and Marine Information Technology Ministry of Education, Xiamen University, Xiamen 361005, China (E-mail: jiahang@stu.xmu.edu.cn, liqun@xmu.edu.cn)}
	\thanks{
		W. Chen is with the Department of Electronic Engineering, Tsinghua University, Beijing 100084, China (E-mail: wchen@tsinghua.edu.cn).}
	\thanks{
		K. B. Letaief is with the Department of Electronic and Computer Engineering, Hong Kong University of Science and Technology, Hong Kong (E-mail: eekhaled@ust.hk). He is also with Peng Cheng Laboratory, Shenzhen, China.}
}

\maketitle

\begin{abstract}
	
Over-the-air computation (AirComp) based federated learning (FL) is capable of achieving fast model aggregation by exploiting the waveform superposition property of multiple access channels.
However, the model aggregation performance is severely limited by the unfavorable wireless propagation channels.
In this paper, we propose to leverage intelligent reflecting surface (IRS) to achieve fast yet reliable model aggregation for AirComp-based FL.
To optimize the learning performance, we formulate an optimization problem that jointly optimizes the device selection, the aggregation beamformer at the base station (BS), and the phase shifts at the IRS to maximize the number of devices participating in the model aggregation of each communication round under certain mean-squared-error (MSE) requirements.
To tackle the formulated highly-intractable problem, we propose a two-step optimization framework.
Specifically, we induce the sparsity of device selection in the first step, followed by solving a series of MSE minimization problems to find the maximum feasible device set in the second step.
We then propose an alternating optimization framework, supported by the difference-of-convex-functions programming algorithm for low-rank optimization, to efficiently design the aggregation beamformers at the BS and phase shifts at the IRS.
Simulation results will demonstrate that our proposed algorithm and the deployment of an IRS can achieve a lower training loss and higher FL prediction accuracy than the baseline algorithms.

\end{abstract}

	
\section{Introduction}

Recent years have witnessed a bloom of artificial intelligence (AI) applications, such as chess play \cite{silver2018general}, natural language generation \cite{gatt2018survey}, and image classification \cite{krizhevsky2012imagenet}.
By adopting advanced machine learning techniques, particularly reinforcement learning and deep learning, computers are able to mimic human behaviours by exploiting tremendous computing power and large amounts of data.
With the further rise of edge computing and Internet of Things (IoT), there emerges a new AI paradigm, named \textit{edge AI} \cite{mao2017mobileedgecomputing, EdgeAI, lataief2019roadmap6g, shi2020edge}, which pushes the AI frontier from the cloud center to the network edge. 
As the data collection and processing are mostly performed at the network edge, the service latency and energy consumption of edge devices can be significantly reduced by edge AI.
As a promising framework for edge AI, \textit{federated learning} (FL) \cite{FedAvg, FLConcept} has recently been proposed to coordinate multiple edge devices to collaboratively train a global AI model. 
Specifically, FL iteratively performs the following two processes \cite{FedAvg}: 1) \textit{model aggregation}: the edge server receives the local model updates from the edge devices over multiple-access channels, and then updates the global model by averaging over the received local model updates; and 2) \textit{model dissemination}: the edge server broadcasts its updated global model to the edge devices, each of which updates the local model based on its own local dataset. 
As only model parameters rather than the real raw data are transmitted to the edge server in the model aggregation process, FL is capable of achieving privacy protection.

As the edge devices are usually connected to the edge server over wireless channels, the model parameters received by the edge server are inevitably distorted by channel fading and additive noise. 
To tackle this issue, several \textit{digital FL} schemes have been proposed to achieve reliable and accurate model aggregation \cite{elgabli2020harnessing, chen2020joint, ren2020scheduling, yuan2020scheduling}.
Specifically, each edge device is allocated an orthogonal resource block to upload its local model parameters, while the edge server is assumed to correctly decode all the local models by adopting the adaptive modulation and coding scheme \cite{elgabli2020harnessing}.
The authors in \cite{chen2020joint} minimized the training loss by jointly optimizing the resource allocation and device selection, taking into account the delay and energy consumption requirements.
Besides, the device scheduling policies for model uploading in each communication round were proposed in \cite{ren2020scheduling} and \cite{yuan2020scheduling} to speed up the convergence rate of FL.
However, the aforementioned studies adopted orthogonal multiple access (OMA) based resource allocation schemes, such as time division multiple access (TDMA) and orthogonal frequency division multiple access (OFDMA), where the required radio resources are linearly scaling with the number of edge devices that participate in FL. When the number of edge devices is large, a substantial communication latency is introduced in the model aggregation process and in turn becomes the performance-limiting factor of FL.

To address the above challenges, \textit{over-the-air computation} (AirComp) empowered \textit{analog FL} emerged to enhance the learning performance under the limited communication bandwidth and stringent latency requirements.
AirComp merges the concurrent data transmission from multiple devices and the function computation via exploiting the waveform superposition property of multiple-access channels \cite{nazer2007computation, chen2018uniform, zhu2019mimo}.
Meanwhile, as the edge server in FL is merely interested in the aggregated model rather than the individual local models, AirComp, as a non-orthogonal multiple access (NOMA) scheme, is recognized as a promising solution for achieving spectral-efficient and low-latency FL \cite{yang2020federated, zhu2020broadband, amiri2020machine, sery2020analog, zhang2020gradient}.
Specifically, the authors in \cite{yang2020federated} proposed a fast model aggregation approach by jointly optimizing device selection and receive beamforming to improve the statistical learning performance under certain mean-squared-error (MSE) requirements for on-device distributed FL.
The authors in \cite{zhu2020broadband} developed a broadband analog aggregation scheme for low-latency FL by considering the communication-and-learning trade-off.
The results in \cite{amiri2020machine} demonstrated that the analog approach via AirComp converges faster than the digital approach due to its more efficient use of limited radio bandwidth.
In \cite{sery2020analog}, the authors developed a gradient-based algorithm to directly deal with noise distorted gradients for FL over wireless channels.
In addition, the authors in \cite{zhang2020gradient} studied the optimal power control problem for AirComp-based FL with gradient statistics.
To achieve an average behavior of local model updates during model aggregation, magnitude alignment should be achieved at the edge server to reduce the aggregation error of AirComp \cite{gao2020optimized, zhu2020over, tang2020reliable}.
However, unfavorable propagation environment inevitably leads to magnitude reduction and misalignment \cite{tang2020reliable}, which in turn degrade the model aggregation accuracy of AirComp-based FL.

To overcome the detrimental effect of channel fading in wireless networks, \textit{intelligent reflecting surface} (IRS) is a cost-effective technology for improving the spectral and energy efficiency via reconfiguring the wireless propagation environment \cite{wu2020towards, wu2019intelligent, ReconfigurableIRS, chen2019intelligent_secure, AirIRS, wang2020wirelesspowered, yang2020flris}.
In particular, a large number of low-cost passive reflecting elements contained in an IRS are capable of adjusting the phase shift of the incident signal, and thus altering the propagation of the reflected signal. 
The signal reflected by IRS can be constructively superposed with the signal over the direct link to boost the received signal power \cite{wu2020towards}.
Due to the passive nature, the power consumption of the IRS is negligible compared with that of the traditional full-duplex amplify-and-forward relay.
In \cite{wu2019intelligent}, an IRS was deployed to minimize the transmit power of the multi-antenna access point (AP) by jointly optimizing active and passive beamforming, while satisfying the signal-to-interference-plus-noise ratio (SINR) constraints.
A joint design of the downlink transmit power and the phase shifts of IRS was developed in \cite{ReconfigurableIRS} to maximize the energy efficiency.
The authors in \cite{chen2019intelligent_secure} utilized the IRS to enhance the physical layer security by jointly optimizing the beamformers at the base station (BS) and the reflecting coefficients at the IRS.
The authors in \cite{AirIRS} leveraged the IRS to minimize the distortion of AirComp in wireless networks.
Moreover, the MSE of aggregated data can be significantly reduced by deploying an IRS in a wireless-powered AirComp network \cite{wang2020wirelesspowered}.
The aforementioned studies demonstrated the potential gains of deploying an IRS in harsh wireless environment, which motivates us to leverage IRS to compensate for magnitude reduction and misalignment of AirComp in FL systems, thereby achieving a lower training loss and higher test accuracy in fewer communication rounds.

\subsection{Contributions}

In this paper, we exploit the advantages of IRS to design a communication-efficient model aggregation scheme for AirComp-based FL systems.
Developing such a scheme to facilitate fast yet reliable model aggregation is challenging.
On one hand, selecting more devices to participate in FL at each communication round is able to simultaneously collect more local model updates, which has a positive impact on the convergence rate of the training process.
On the other hand, selecting more devices in each communication round enlarges the model aggregation error due to the inevitable magnitude misalignment at the edge server, which is detrimental to the convergence rate of the training process.
As a result, the edge devices should be appropriately selected to speed up the overall convergence rate of FL.
The optimization of device selection can be achieved by solving an $\ell_0$-norm minimization problem, which, however, requires the design of an efficient algorithm to accurately induce sparsity.
In addition, to prevent the learning performance from being degraded by the large model aggregation error, 
it is critical yet challenging to optimize the aggregation beamformer at the BS and the phase shifts at the IRS to combat severe channel fading and reduce the impact of additive noise.
Therefore, it is necessary to jointly optimize the device selection, the aggregation beamformer at the BS, and the phase shifts at the IRS to improve the learning efficiency and the prediction accuracy of the IRS-assisted AirComp-based FL system.
The main contributions of this paper are summarized as follows.
\begin{itemize}
	\item We propose an IRS-assisted AirComp-based FL system that is able to achieve fast yet reliable model aggregation. In particular, an IRS is deployed to mitigate the magnitude misalignment at the edge server during model aggregation, so as to schedule more edge devices to participate in FL at each communication round under a certain MSE requirement of each aggregated model, thereby achieving a lower training loss and higher test accuracy in fewer communication rounds. 
	 
	\item We propose to jointly optimize the device selection, the aggregation beamformer at the BS, and the phase shifts at the IRS, which, however, is highly intractable due to the sparse objective function as well as the biquadratic constraints due to the coupling between the aggregation beamformer at the BS and the phase shifts at the IRS.
	
	\item We first propose a two-step optimization framework to tackle the sparse objective function. Specifically, we induce the sparsity of the device selection by adopting $\ell_1$-relaxation for the $\ell_0$-norm objective function in the first step, followed by solving a series of MSE minimization problems in the second step to find the maximum feasible device selection set. Then, we propose an alternating optimization method to decouple the aggregation beamformer at the BS and the phase shifts at the IRS, thereby removing the obstacles caused by the biquadratic constraints in our problem.
	
	\item To address the nonconvex quadratic constraints in each subproblem resulting from the alternating optimization method, we convert them into rank-one constrained semidefinite programming (SDP) problems via matrix lifting. Subsequently, we reformulate the SDP problems as difference-of-convex (DC) programming problems by introducing DC representations for rank-one constraints, so as to effectively solve the low-rank optimization problem with the proposed DC algorithm. 
\end{itemize}
 
Simulation results demonstrate that the proposed IRS-assisted AirComp-based FL system is able to schedule more devices in each communication round under certain MSE requirements. The proposed two-step alternating DC algorithm achieves more accurate feasible set detection than the SDR approach.
Moreover, our proposed algorithm enables FL to converge faster and achieve more accurate prediction in the experiment of training a deep convolutional neural network (CNN) on the MNIST dataset \cite{lecun1998mnist} than other baseline schemes.

\subsection{Organization and Notations}

The rest of this paper is organized as follows.
Section \ref{SecSystem} describes the system model and problem formulation in IRS-assisted FL system.
In Section \ref{SecAlternating}, we propose a two-step framework to solve the problem.
Section \ref{SecDC} presents a two-step alternating DC algorithm for solving the problem.
The simulation results are provided in Section \ref{SecSimulation}.
Finally, Section \ref{SecConclusion} concludes this work.

Italic, boldface lower-case, and boldface upper-case letters denote scalar, vector, and matrix, respectively.
$\mathbb{R}^{m \times n}$ and $\mathbb{C}^{m \times n}$ denote the real and complex domain with the space of $m \times n$, respectively.
The operators $(\cdot)^\mathsf{T}$, $(\cdot)^\mathsf{H}$, $\mathrm{tr}(\cdot)$, and $\mathrm{diag}(\cdot)$ denote the transpose, Hermitian transpose, trace, and diagonal matrix, respectively. $\mathbb{E}[\cdot]$ denotes the statistical expectation. The operator $|\cdot|$ denotes the cardinality of a set or the absolute value of a scalar number, and $\|\cdot\|$ denotes the Euclidean norm.

\section{System Model and Problem Formulation} \label{SecSystem}

In this section, we develop a computation and communication co-design for fast and reliable model aggregation in AirComp-based FL systems, where an IRS is deployed to compensate for the magnitude reduction and misalignment of AirComp. 

\subsection{FL Model} \label{SecSystem_FL_model}

The IRS-assisted AirComp-based FL system under consideration consists of one $M$-antenna BS serving as an edge server, $K$ single-antenna edge devices, and an IRS with $N$ passive reflecting elements, as shown in Fig.~\ref{fig_system}. 
Edge device $k \in \mathcal{K} = \{1, 2, \dots, K\}$ has its own local dataset $\mathcal{D}_k$ with $D_k = |\mathcal{D}_k|$ labeled data samples $\{(\bm{u}_i, v_i)\}_{i = 1}^{D_k} \in \mathcal{D}_k$, where $(\bm{u}_i, v_i)$ denotes the input-output data pair consisting of training sample $\bm{u}_i$ and its ground-truth label $v_i$. 
For a given $d$-dimensional model parameter $\bm{z} \in \mathbb{R}^d$, the local loss function for device $k$  is defined as
\begin{align}
	F_k(\bm{z}) = \frac{1}{D_k} \sum_{(\bm{u}_i, v_i) \in \mathcal{D}_k} f(\bm{z}; \bm{u}_i, v_i),
\end{align}
where $f(\bm{z}; \bm{u}_i, v_i)$ denotes the sample-wise loss function.
Without loss of generality, we assume that all local datasets have a uniform size, i.e., $D_k = D, \, \forall \, k \in \mathcal{K}$, as in \cite{zhu2020broadband}.
Then, the global loss function with model parameter $\bm{z}$ can be represented as
\begin{align}
	F(\bm{z}) = \frac{1}{\sum_{k = 1}^K D_k} \sum_{k = 1}^K D_k F_k(\bm{z}) = \frac{1}{K} \sum_{k = 1}^K F_k(\bm{z}).
\end{align}
The learning process aims to optimize the model parameter $\bm{z}$ that minimizes the global loss function, i.e.,
\begin{align}
	\bm{z}^\star = \mathrm{arg} \underset{\bm{z} \in \mathbb{R}^d}{\mathrm{min}} \, F(\bm{z}).
\end{align}
To achieve this purpose, with the traditional method, the BS gathers all the local data from the edge devices to train a global model, which, however, not only increases the computation burden on the centralized server but also causes the privacy concern of edge devices.

\begin{figure}[t]
	\centering
	\includegraphics[scale=0.7]{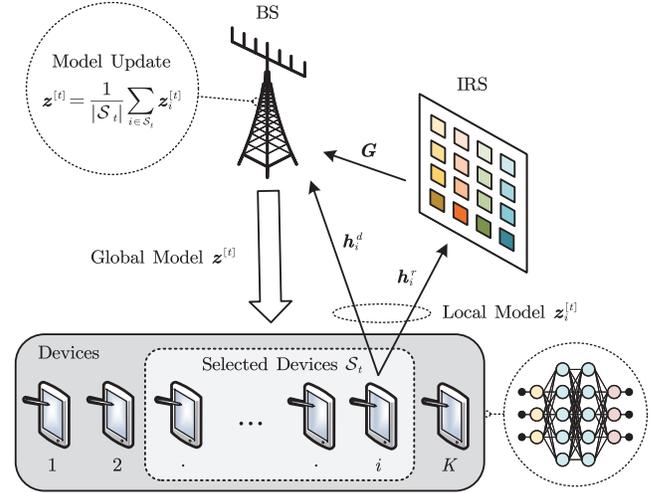}
	\caption{Illustration of an IRS-assisted AirComp-based FL system.}
	\label{fig_system}
\end{figure}

Fortunately, as an on-device distributed machine learning method, FL is able to collaboratively train a global model by coordinating the distributed edge devices to update the local model parameters according to the locally owned training data. Without the need of uploading the local data to the BS, this distributed learning method possesses the advantages of low latency, low power consumption, and high data privacy. 
In this paper, we leverage FedAvg \cite{FedAvg}, also referred to as model averaging, to train a global model.
Specifically, at the $t$-th communication round, the BS and the edge devices perform the following procedures
\begin{itemize}
	\item The BS broadcasts the current global model $\bm{z}^{[t-1]}$ to the edge devices belonging to a selected set, denoted as $\mathcal{S}_t \subseteq \mathcal{K}$.
	
	\item Based on the received global model $\bm{z}^{[t-1]}$, each edge device $i \in \mathcal{S}_t$ performs a local model update algorithm by utilizing its local dataset ${\mathcal{D}_i}$ to obtain an updated local model $\bm{z}_i^{[t]}$.
	
	\item All the local model updates are aggregated at the BS by taking an average to obtain the updated global model $\bm{z}^{[t]}$, which is given by
	\begin{align} \label{Aggregation}
		\bm{z}^{[t]} = \frac{1}{|\mathcal{S}_t|} \sum_{i \in \mathcal{S}_t} \bm{z}_i^{[t]}.
	\end{align}
\end{itemize}

In the following, we train a deep CNN on the MNIST dataset by using the FedAvg algorithm to show the impact of the number of selected devices on the training loss and the test accuracy under different model aggregation errors. The aggregated global model is given by
\begin{align}
	\hat{\bm{z}} = \frac{1}{|\mathcal{S}|} \sum_{i \in \mathcal{S}} \bm{z}_i + \bm{e},
\end{align}
where $\bm{e} \sim \mathcal{N}(0, \sigma_0^2 \bm{I})$.
As shown in Fig.~\ref{fig_training_loss_clients} and Fig.~\ref{fig_test_accuracy_clients}, selecting more devices to participate in the training process is able to obtain a model that provides lower training loss and higher test accuracy.
Besides, the training loss increases and the test accuracy decreases as the model aggregation error increases under the same number of selected devices. 
Therefore, it is critical to schedule more devices and reduce the aggregation error in each communication round for training a high quality model.
Note that the aggregation error is mainly caused by channel fading and additive noise during model aggregation and dissemination over wireless channels.
Motivated by these observations, we propose to jointly optimize the device selection, the aggregation beamformer at the BS, and the phase shifts at the IRS to maximize the number of edge devices participating in model aggregation while ensuring the aggregation error is within a certain bound.
Such a design is capable of enhancing the performance of FL in terms of the training loss and test accuracy.

\begin{figure}[t]
	\centering
	\includegraphics[scale=0.55]{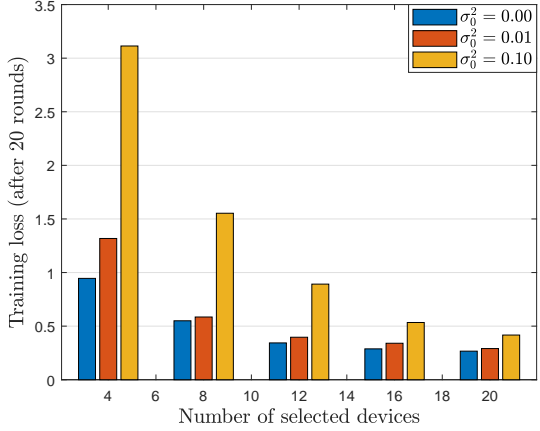}
	\caption{Training loss versus the number of selected devices under different model aggregation errors.}
	\label{fig_training_loss_clients}
\end{figure}

\begin{figure}[t]
	\centering
	\includegraphics[scale=0.55]{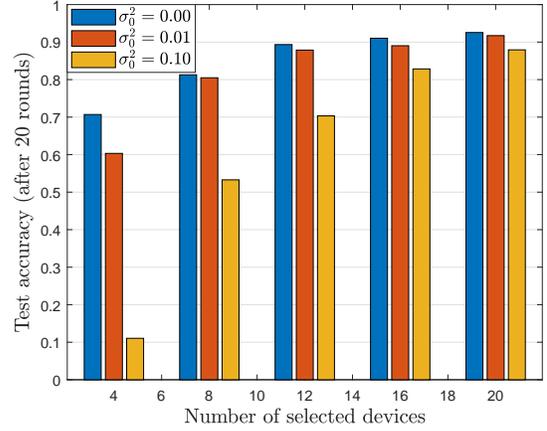}
	\caption{Test accuracy versus the number of selected devices under different model aggregation errors.}
	\label{fig_test_accuracy_clients}
\end{figure}

\subsection{Communication Model for IRS-Assisted AirComp}

Since the average sum in Eq.~\eqref{Aggregation} for model aggregation falls into the category of nomographic functions \cite{zhu2019mimo}, AirComp as a promising technique can be utilized to enhance the efficiency of model aggregation from distributed edge devices.
Let ${\phi _i}(\bm{x}) = \bm{x}$ denote the pre-processing function at device $i$ and $\psi(\bm{x}) = \frac{1}{|\mathcal{S}|} \bm{x}$ denote the post-processing function at the BS.
The target function for aggregating the local model updates at the BS can be expressed as
\begin{align}
\bm{z} = \psi \left( \sum\limits_{i \in \mathcal{S}} \phi _i (\bm{z}_i) \right),
\end{align}
where $\mathcal{S}$ is the device selection set.
We denote $\bm{s}_i = \bm{z}_i \in \mathbb{R}^d$ as the transmit symbol vector at device $i$.
Without loss of generality, the transmit symbols are assumed to be independent and normalized to have zero mean and unit variance, i.e., $\mathbb{E}[{\bm{s}_i}\bm{s}_i^\mathsf{H}] = \bm{I}_d$ \cite{zhu2020broadband}.
Due to the limited capacity for data storage and computing at the edge devices, the dimension of model parameters is set to ensure that the entire model parameters can be transmitted within one transmission interval \cite{liu2020privacy}.
To simplify the notation, let $\bar{s}_i$ denote a typical entry of $\bm{s}_i$ within one communication interval.
The target function to be estimated at the BS is given by
\begin{align}
g = \sum_{i \in \mathcal{S}} \phi_i(\bar{s}_i) = \sum_{i \in \mathcal{S}} \bar{s}_i.
\end{align}
The transmitted signals may encounter detrimental channel conditions during the model aggregation process through AirComp in the uplink, which leads to magnitude reduction and misalignment, thereby enlarging the aggregation error at the BS.
To tackle this issue, we propose to deploy an IRS to alleviate the distortion of AirComp.


Let $\bm{h}_i^\mathrm{d} \in {\mathbb{C}^M}$, $\bm{h}_i^\mathrm{r} \in {\mathbb{C}^N}$, and $\bm{G} \in {\mathbb{C}^{M \times N}}$ denote the channel responses from device $i$ to the BS, from device $i$ to the IRS, and from the IRS to the BS, respectively.
The channel gain of each link is assumed to be invariant within one transmission interval.
In addition, with various channel estimation methods proposed for IRS-assisted wireless networks \cite{mishra2019channel, zheng2020intelligent, wang2020channel}, we assume that the perfect CSI is available in this paper, as in \cite{wu2020towards, wu2019intelligent, ReconfigurableIRS, chen2019intelligent_secure, AirIRS, wang2020wirelesspowered}.
The diagonal phase-shift matrix of the IRS is denoted by $\bm{\Theta}  = {\textrm{diag}}({\beta e^{j{\theta _1}}}, \dots, {\beta e^{j{\theta _N}}}) \in {\mathbb{C}^{N \times N}}$, where ${\theta _n} \in [0, 2\pi)$ denotes the phase shift of element $n$ and $\beta \in [0, 1]$ is the amplitude reflection coefficient on the incident signals. Without loss of generality, we assume $\beta = 1$ in this paper \cite{wu2019intelligent}.
Compounded with reflected signals, the received signal at the BS is given by
\begin{align}
\bm{y} = \sum\limits_{i \in \mathcal{S}} {(\bm{G\Theta h}_i^\mathrm{r} + \bm{h}_i^\mathrm{d}){w_i}{\bar{s}_i}}  + \bm{n},
\end{align}
where ${w_i} \in \mathbb{C}$ is the transmit scalar of device $i$ and $\bm{n}\sim\mathcal{CN}\left( {\bm{0},{\sigma ^2}\bm{I}} \right)$ is the additive white Gaussian noise (AWGN).

By denoting the aggregation beamforming vector at the BS as $\bm{m} \in {\mathbb{C}^M}$, the estimated target function before post-processing can be expressed as
\begin{align}
\hat g = \frac{1}{\sqrt{\eta}}{\bm{m}^\mathsf{H}}\bm{y} = \frac{1}{{\sqrt \eta  }}{\bm{m}^\mathsf{H}}\sum\limits_{i \in \mathcal{S}} {(\bm{G\Theta h}_i^\mathrm{r} + \bm{h}_i^\mathrm{d}){w_i}{\bar{s}_i}}  + \frac{1}{\sqrt \eta }{\bm{m}^\mathsf{H}}\bm{n},
\end{align}
where $\eta$ is a denoising factor. Thus, we can obtain the aggregated global model at the BS by post-processing $\hat z=\psi(\hat g)$.
The distortion of the estimated aggregated model, which quantifies the performance for global model aggregation via AirComp, can be measured by the MSE between $\hat g$ and the target value $g$ as follows
\begin{align}\label{equMSE}
& \mathsf{MSE}\left(\hat{g}, g\right) = \mathbb{E}\left[{\left| {\hat g - g} \right|^2}\right] \notag\\
&= \sum\limits_{i \in \mathcal{S}} {{\left| \frac{1}{\sqrt \eta }{\bm{m}^\mathsf{H}} (\bm{G\Theta h}_i^\mathrm{r} + \bm{h}_i^\mathrm{d}) {w_i} - 1\right|}^2}  + \frac{{\sigma ^2}{{\left\|\bm{m}\right\|}^2}}{\eta }.
\end{align}
The following proposition presents the optimal transmit scalars at the edge devices to minimize the MSE.
\begin{prop} \label{transmitter}
Given the aggregation beamforming vector $\bm{m}$ and the phase-shift matrix $\bm{\Theta}$, the minimum MSE is obtained by using the following optimal transmit scalar
\begin{align} \label{Optimalwi}
w_i^\star = \sqrt{\eta} \frac{(\bm{m}^\mathsf{H}(\bm{G} \bm{\Theta} \bm{h}_i^\mathrm{r} + \bm{h}_i^\mathrm{d}))^\mathsf{H}}{|\bm{m}^\mathsf{H} (\bm{G} \bm{\Theta} \bm{h}_i^\mathrm{r} + \bm{h}_i^\mathrm{d})|^2}, \, \forall \, i \in \mathcal{S}.
\end{align}
\end{prop}

\begin{proof}
Please refer to Appendix \ref{ap_transmitter}.
\end{proof}

The transmit power of device $i$ is constrained by a given maximum transmit power ${P_0}>0$, i.e., ${\left| {{w_i}} \right|^2} \le {P_0}$.
With the optimal transmit scalar $w_i^\star$ given in \eqref{Optimalwi}, we have
\begin{align}
\eta = P_0 \, \underset{i \in \mathcal{S}}{\mathrm{min}} \, |\bm{m}^\mathsf{H} (\bm{G} \bm{\Theta} \bm{h}_i^\mathrm{r} + \bm{h}_i^\mathrm{d})|^2.
\end{align}
Therefore, the minimum MSE is given by
\begin{align}
\mathsf{MSE}\left(\hat{g}, g\right) = \frac{\sigma^2}{P_0} \, \underset{i \in \mathcal{S}}{\mathrm{max}} \, \frac{\|\bm{m}\|^2}{{|{{\bm{m}^\mathsf{H}}(\bm{G\Theta h}_i^\mathrm{r} + \bm{h}_i^\mathrm{d})}|}^2}.
\end{align}

\subsection{Problem Formulation}


As observed in Section \ref{SecSystem_FL_model}, we aim to maximize the number of selected devices while satisfying the MSE requirement of model aggregation to speed up the convergence of the training process and to avoid the notable reduction of the prediction accuracy. Specifically, given the MSE requirement $\gamma > 0$ for model aggregation, the corresponding optimization problem can be formulated as
\begin{subequations}
\begin{align}
\underset{\mathcal{S}, \bm{m}, \bm{\Theta}}{\text{maximize}} \hspace{3mm} & |\mathcal{S}| \\
\text{subject to} \hspace{3mm} & {\mathop {\max }\limits_{i \in \mathcal{S}} \frac{{\| \bm{m} \|}^2}{{| {{\bm{m}^\mathsf{H}}(\bm{G\Theta}\bm{h}_i^\mathrm{r} + \bm{h}_i^\mathrm{d})} |}^2}} \le \gamma, \label{consMSE}\\
& |\bm{\Theta}_{n, n}| = 1, \, \forall \, n \in \{1, \dots, N\}.\label{constheta}
\end{align}
\end{subequations}
To facilitate the algorithm design, the MSE constraint \eqref{consMSE} can be rewritten as nonconvex constraints with quadratic and biquadratic terms, as presented in Proposition \ref{QuadCons}.
\begin{prop} \label{QuadCons}
The constraint \eqref{consMSE} can be equivalently rewritten as the following constraints:
\begin{align} \label{consquad}
	\|\bm{m}\|^2 - \gamma |\bm{m}^\mathsf{H} (\bm{G} \bm{\Theta} \bm{h}_i^\mathrm{r} + \bm{h}_i^\mathrm{d})|^2 \le 0, i \in \mathcal{S},
\end{align}
where ${\left\| \bm{m} \right\|^2} \ge 1$.
\end{prop}
\begin{proof}
Please refer to Appendix \ref{ap_quadcons}.
\end{proof}

According to Proposition \ref{QuadCons}, the objective function $\left| \mathcal{S} \right|$ represents the number of feasible MSE constraints \eqref{consquad}, which should be maximized under the regularity condition $\|\bm{m}\|^2 \ge 1$.
By adding an auxiliary variable $\bm{x}$ \cite{SmoothedLp}, we equivalently transform the problem of maximizing the number of feasible MSE constraints into the problem of minimizing the number of nonzero $x_i$'s.
Hence, we turn to solve the following sparse optimization problem
\begin{subequations}
\begin{align}
\mathscr{P}:
\underset{\bm{x} \in \mathbb{R}_+^K, \bm{m}, \bm{\Theta}}{\text{minimize}} {\ } & \|\bm{x}\|_0 \\
\text{subject to} {\ } & \|\bm{m}\|^2 - \gamma |\bm{m}^\mathsf{H} (\bm{G} \bm{\Theta} \bm{h}_i^\mathrm{r} + \bm{h}_i^\mathrm{d})|^2 \le x_i, \, \forall \, i \in \mathcal{K}, \label{consPmse} \\
& \|\bm{m}\|^2 \ge 1, \label{consPm} \\
& |\bm{\Theta}_{n, n}| = 1, \, \forall \, n \in \{1, \dots, N\}. \label{consPtheta}
\end{align}
\end{subequations}

Note that the selection of each edge device is indicated by the sparsity structure of $\bm{x}$, i.e., $x_i = 0$ indicates that device $i$ can be selected while satisfying the MSE requirement.
Due to the sparse objective function and nonconvex constraints with biquadratic \eqref{consPmse} and quadratic \eqref{consPm} terms, problem $\mathscr{P}$ is computationally difficult. 
To tackle this issue, we shall propose a two-step alternating low-rank optimization framework in the following section.

\section{Alternating Low-Rank Optimization Framework for Model Aggregation}\label{SecAlternating}

In this section, we propose a two-step framework to solve problem $\mathscr{P}$ for IRS-assisted AirComp-based FL with device selection,
followed by proposing to use the alternating optimization approach to solve the problem in each step.



\subsection{Proposed Two-Step Framework for Solving Problem $\mathscr{P}$} \label{SubsecTwostep}

The main idea of our proposed two-step framework is to induce the sparsity of $\bm{x}$ in the first step, so as to determine the priority for each device to be selected. With the obtained priority vector, we then solve a series of MSE minimization problems to find the maximum feasible device set while satisfying the MSE requirement in the second step.

\subsubsection{Sparsity Inducing}

For the nonconvex sparse objective function being in the form of $\ell_0$-norm, we adopt the well-recognized $\ell_1$-norm as a convex surrogate \cite{ConOp}.
To solve problem $\mathscr{P}$, we shall solve the following problem in the first step:
\begin{align}
\mathscr{P}_1:
\underset{\bm{x} \in \mathbb{R}_+^K, \bm{m}, \bm{\Theta}}{\text{minimize}} \hspace{3mm} & \|\bm{x}\|_1 \notag \\
\text{subject to} \hspace{3mm} & \mathrm{constraints \, \eqref{consPmse}, \, \eqref{consPm}, \, \eqref{consPtheta}}.
\end{align}
After solving problem $\mathscr{P}_1$, we proceed to the second step to check the feasibility of the selected devices and find the maximum number of edge devices under the MSE constraint.

\subsubsection{Feasibility Detection}

The value of $x_i$ obtained from the first step characterizes the disparity between the MSE requirement and the achievable MSE for device $i$.
Therefore, the smaller the value of $x_i$, the higher priority device $i$ being selected in the second step.
We sort $\{x_i\}_{i = 1}^K$ in an ascending order ${x_{\pi (1)}} \le \dots \le {x_{\pi (K)}}$ to determine the priority of edge devices, where $x_{\pi(i)}$ denotes the $i$-th smallest element in $\{x_i\}_{i = 1}^K$.
We adopt the bisection method to find the maximum value of $k$ that enables all devices in the set $\mathcal{S}^{[k]} = \{\pi(1), \pi(2), \dots ,\pi(k)\}$ to be feasibly selected.
Specifically, for a given device set $\mathcal{S}^{[k]}$,
we check the feasibility via comparing the MSE requirement with the minimum maximal MSE of selected devices in $\mathcal{S}^{[k]}$ obtained from the following problem:
\begin{subequations} \label{ProbS2}
\begin{align} 
\mathscr{P}_2:
\underset{\bm{m}, \bm{\Theta}}{\text{minimize}} \hspace{3mm} & \underset{i \in \mathcal{S}^{[k]}}{\mathrm{max}} {\ } \frac{\|\bm{m}\|^2}{|\bm{m}^\mathsf{H} (\bm{G\Theta h}_i^\mathrm{r} + \bm{h}_i^\mathrm{d})|^2} \label{objfuncPS2} \\
\text{subject to} \hspace{3mm} &|\bm{\Theta}_{n, n}| = 1, \, \forall \, n \in \{1, \dots, N\}.
\end{align}
\end{subequations}
If the optimal objective value of problem \eqref{ProbS2} is less than the required MSE, then set $\mathcal{S}^{[k]}$ is considered as a feasible set.

\subsection{Alternating Low-Rank Optimization} \label{SecAlter_LowRank}

It can be observed that constraint \eqref{consPmse} and objective function \eqref{objfuncPS2} are both nonconvex due to the coupled optimization variables. To address this issue, we propose to apply alternating optimization \cite{AirIRS}.

\subsubsection{Sparsity Inducing}

In the first step, variables $(\bm{x}, \bm{m})$ and $\bm{\Theta}$ of problem $\mathscr{P}_1$ can be optimized alternately. Specifically, when the phase-shift matrix $\bm{\Theta}$ is fixed (i.e., the combined channel vector $\bm{h}_i = \bm{G} \bm{\Theta} \bm{h}_i^\mathrm{r} + \bm{h}_i^\mathrm{d}$ between device $i$ and the BS is fixed), the problem can be expressed as
\begin{align} \label{ProbS1.1}
\underset{\bm{x} \in \mathbb{R}_+^K, \bm{m}}{\text{minimize}} \hspace{3mm} & \|\bm{x}\|_1 \notag \\
\text{subject to} \hspace{3mm} & \mathrm{constraints \, \eqref{consPmse}, \, \eqref{consPm}}.
\end{align}
To address the nonconvexity of biquadratic and quadratic constraints \eqref{consPmse} and \eqref{consPm}, we further transform problem \eqref{ProbS1.1} into an SDP problem via the matrix lifting technique \cite{MatrixLift}. By denoting $\bm{M} = \bm{m} \bm{m}^\mathsf{H}$, problem \eqref{ProbS1.1} can be rewritten as a low-rank optimization problem:
\begin{align}
	\mathscr{P}_{1, 1}:
	\underset{\bm{x} \in \mathbb{R}_+^K, \bm{M}}{\text{minimize}} \hspace{3mm} & \|\bm{x}\|_1 \notag \\
	\text{subject to} \hspace{3mm} & \mathrm{tr}(\bm{M}) - \gamma \cdot \mathrm{tr}(\bm{M} \bm{H}_i) \le x_i, \, \forall \, i \in \mathcal{K}, \notag \\
	& \mathrm{tr}(\bm{M}) \ge 1, \notag \\
	& \bm{M} \succeq \bm{0}, \, \mathrm{rank}(\bm{M}) = 1,
\end{align}
where $\bm{H}_i = \bm{h}_i \bm{h}_i^\mathsf{H}$.

On the other hand, when the auxiliary vector $\bm{x}$ and the aggregation beamforming vector $\bm{m}$ are fixed, problem $\mathscr{P}_1$ is reduced to be a feasibility detection problem of phase-shift matrix $\bm{\Theta}$.
By denoting $\bm{v} = [e^{j \theta_1}, \dots, e^{j \theta_N}]^\mathsf{T}$, $\bm{a}_i^\mathsf{H} = \bm{m}^\mathsf{H} \bm{G} \mathrm{diag}(\bm{h}_i^\mathrm{r})$, and $c_i = \bm{m}^\mathsf{H} \bm{h}_i^\mathrm{d}$, the problem can be expressed as
\begin{subequations} \label{ProbS1.2}
\begin{align}
\text{find} \hspace{3mm} & \bm{v} \\
\text{subject to} \hspace{3mm} & \|\bm{m}\|^2 - \gamma |\bm{a}^\mathsf{H} \bm{v} + c_i|^2 \le x_i, \, \forall \, i \in \mathcal{K}, \label{consPS1.2mse} \\
& |v_n| = 1, \, \forall \, n \in \{1, \dots, N\}.
\end{align}
\end{subequations}
We denote $\bar{\bm{v}} = [\bm{v}, t]^\mathsf{T}$ by introducing an auxiliary variable $t$. Constraints \eqref{consPS1.2mse} can be rewritten as
\begin{align}
\|\bm{m}\|^2 - \gamma \left(\bar{\bm{v}}^\mathsf{H} \bm{R}_i \bar{\bm{v}} + |c_i|^2\right) \le x_i, \, \forall \, i \in \mathcal{K},
\end{align}
where 
\begin{align}
\bm{R}_i = \left[\begin{array}{cc}
	\bm{a}_i \bm{a}_i^\mathsf{H} & \bm{a}_i c_i \\
	c_i^\mathsf{H} \bm{a}_i^\mathsf{H} & 0
\end{array} \right].
\end{align}
Since $\bar{\bm{v}}^\mathsf{H} \bm{R}_i \bar{\bm{v}} = \mathrm{tr}(\bm{R}_i \bar{\bm{v}} \bar{\bm{v}}^\mathsf{H})$, we lift $\bar{\bm{v}}$ as a positive semidefinite (PSD) matrix $\bm{V} = \bar{\bm{v}} \bar{\bm{v}}^\mathsf{H}$ with $\mathrm{rank}(\bm{V}) = 1$. Problem \eqref{ProbS1.2} can be equivalently reformulated as the following low-rank matrix optimization problem:
\begin{align}
\mathscr{P}_{1, 2}:
\text{find} \hspace{3mm} & \bm{V} \notag \\
\text{subject to} \hspace{3mm} & \|\bm{m}\|^2 - \gamma \left(\mathrm{tr}(\bm{R}_i \bm{V}) + |c_i|^2\right) \le x_i, \, \forall \, i \in \mathcal{K}, \notag \\
& \bm{V}_{n, n} = 1, \, \forall \, n \in \{1, \dots, N + 1\}, \notag \\
& \bm{V} \succeq \bm{0}, \, \mathrm{rank}(\bm{V}) = 1.
\end{align}

\subsubsection{Feasibility Detection}

In the second step, we first reformulate problem $\mathscr{P}_2$ as the following problem \cite{AirIRS}:
\begin{align}
	\underset{\bm{m}, \bm{\Theta}}{\text{minimize}} \hspace{3mm} & \|\bm{m}\|^2 \notag \\
	\text{subject to} \hspace{3mm} & |\bm{m}^\mathsf{H} (\bm{G} \bm{\Theta} \bm{h}_i^\mathrm{r} + \bm{h}_i^\mathrm{d})|^2 \ge 1, \, \forall \, i \in \mathcal{S}^{[k]}, \notag \\
	& |\bm{\Theta}_{n, n}| = 1, \, \forall \, n \in \{1, \dots, N\}.
\end{align}
To decouple the optimization variables, we optimize the aggregation beamforming vector $\bm{m}$ and the phase-shift matrix $\bm{\Theta}$ alternately.
Specifically, given the phase-shift matrix $\bm{\Theta}$, we have
\begin{align}
	\underset{\bm{m}}{\text{minimize}} \hspace{3mm} & \|\bm{m}\|^2 \notag \\
	\text{subject to} \hspace{3mm} & |\bm{m}^\mathsf{H} \bm{h}_i|^2 \ge 1, \, \forall \, i \in \mathcal{S}^{[k]}.
\end{align}
This problem can be further represented as a low-rank matrix optimization problem:
\begin{align}
\mathscr{P}_{2, 1}:
\underset{\bm{M}}{\text{minimize}} \hspace{3mm} & \mathrm{tr}(\bm{M}) \notag \\
\text{subject to} \hspace{3mm} & \mathrm{tr}(\bm{M} \bm{H}_i) \ge 1, \, \forall \, i \in \mathcal{S}^{[k]}, \notag \\
& \bm{M} \succeq \bm{0}, \, \mathrm{rank}(\bm{M}) = 1.
\end{align}

On the other hand, given the aggregation beamforming vector $\bm{m}$, we have
\begin{align}
	\text{find} \hspace{3mm} & \bm{v} \notag \\
	\text{subject to} \hspace{3mm} & |\bm{a}^\mathsf{H} \bm{v} + c_i|^2 \ge 1, \, \forall \, i \in \mathcal{S}^{[k]}, \notag \\
	& |v_n| = 1, \, \forall \, n \in \{1, \dots, N\},
\end{align}
and its corresponding low-rank matrix optimization problem is given by
\begin{align}
\mathscr{P}_{2, 2}:
\text{find} \hspace{3mm} & \bm{V} \notag \\
\text{subject to} \hspace{3mm} & \mathrm{tr}(\bm{R}_i \bm{V}) + |c_i|^2 \ge 1, \, \forall \, i \in \mathcal{S}^{[k]}, \notag \\
& \bm{V}_{n, n} = 1, \, \forall \, n \in \{1, \dots, N + 1\}, \notag \\
& \bm{V} \succeq \bm{0}, \, \mathrm{rank}(\bm{V}) = 1.
\end{align}

In summary, the entire proposed two-step alternating DC algorithm for solving the sparse and low-rank optimization problem $\mathscr{P}$ is presented in Algorithm \ref{alg_two_step_alternating_dc}.
The resulting problems $\mathscr{P}_{1, 1}$, $\mathscr{P}_{1, 2}$, $\mathscr{P}_{2, 1}$, and $\mathscr{P}_{2, 2}$ in the alternating low-rank optimization are still nonconvex because of the fixed rank-one constraints.
This nonconvexity issue can be tackled by simply dropping the nonconvex rank constraints via the SDR technique \cite{MatrixLift}.
In the procedure of solving the relaxed SDP problems, if the obtained solution fails to be rank-one, the Gaussian randomization method \cite{MatrixLift} can be adopted to obtain a suboptimal solution.
However, if the number of antennas and the number of reflecting elements are large, the performance of the SDR technique degenerates in the resulting high-dimensional optimization problems due to the low probability of returning rank-one solutions \cite{chen2018uniform}.
To address the limitations of the SDR technique, we present a novel DC programming approach for inducing rank-one solutions in the next section.

\begin{algorithm}[t]
	\caption{Two-step alternating DC algorithm for solving problem $\mathscr{P}$ in FL with device selection.}
	\label{alg_two_step_alternating_dc}
	
	\textbf{Step 1:} Sparsity Inducing \\
	\KwIn{Initial point $\bm{\Theta}^0$ and predefined threshold $\epsilon > 0$.}
	\For{$t \leftarrow 1, 2, \dots$}{
		Given $\bm{\Theta}^{t - 1}$, obtain solution $(\bm{x}^t, \bm{m}^t)$ by solving problem $\mathscr{P}_{1, 1}$. \\
		Given $(\bm{x}^t, \bm{m}^t)$, obtain solution $\bm{\Theta}^t$ by solving problem $\mathscr{P}_{1, 2}$. \\
		\If{Decrease of the objective value of problem $\mathscr{P}_1$ is below $\epsilon$}{
			\textbf{break}.
		}
	}
	\KwOut{$\bm{x}^\star \leftarrow \bm{x}^t$.}
	~ \\
	\textbf{Step 2:} Feasibility Detection \\
	\KwIn{Set $\mathcal{S}^{[K]} = \{\pi(1), \pi(2), \dots, \pi(K)\}$ obtained by ordering $\bm{x}^\star$ in an ascending order as $x_{\pi(1)} \le \dots \le x_{\pi(K)}$, $N_\mathrm{low} \leftarrow 0$, $N_\mathrm{up} \leftarrow K$, and predefined threshold $\epsilon > 0$.}
	$k \leftarrow K$. \\
	\While{$N_\mathrm{up} - N_\mathrm{low} > 1$}{
		Initialize $\bm{\Theta}^0$ and ${\mathcal{S}^{[k]}}\leftarrow\left\{ {\pi (1),\pi (2), \dots ,\pi (k)} \right\}$. \\
		\For{$t \leftarrow 1, 2, \dots$}{
			Given $\bm{\Theta}^{t - 1}$, obtain solution $\bm{m}^t$ by solving problem $\mathscr{P}_{2, 1}$. \\
			\uIf{Maximum $\mathsf{MSE} \le \gamma$}{
				$N_\mathrm{low} \leftarrow k$. \\
				$\bar{\bm{m}} \leftarrow \bm{m}^t$. \\
				$k \leftarrow \lfloor\frac{N_\mathrm{low} + N_\mathrm{up}}{2}\rfloor$. \\
				\textbf{Break}. \\
			}
			~ \\
			\ElseIf{Decrease of the objective value of problem $\mathscr{P}_2$ is below $\epsilon$}{
				$N_\mathrm{up} \leftarrow k$. \\
				$k \leftarrow \lfloor\frac{N_\mathrm{low} + N_\mathrm{up}}{2}\rfloor$. \\
				\textbf{Break}. \\
			}
			Given $\bm{m}^t$, obtain solution $\bm{\Theta}^t$ by solving problem $\mathscr{P}_{2, 2}$.
		}
	}
	\KwOut{$\bm{m}^\star  \leftarrow \bar{\bm{m}}$ and the set of selected devices $\mathcal{S}^{[k]} \leftarrow \{\pi(1), \pi(2), \dots, \pi(k^\star)\}$ with $k^\star \leftarrow N_\mathrm{low}$.}
\end{algorithm}


\section{Alternating DC Approach For Low-Rank Optimization}\label{SecDC}

In this section, we present a DC formulation for the rank-one constrained SDP problems in the alternating procedure, followed by proposing a two-step alternating DC algorithm to solve problem $\mathscr{P}$ in IRS-assisted AirComp-based FL systems.

\subsection{DC Formulation for Rank-One Constrained Problems}

The accurate detection of the rank-one constraint plays a critical role in precisely detecting the feasibility of nonconvex quadratic constraints, which is important in our two-step framework for device selection.
Therefore, we provide a DC representation for the rank-one constraints in the aforementioned problems in Section \ref{SecAlter_LowRank}.

The rank-one constraint of PSD matrix $\bm{M} \in \mathbb{C}^{M\times M}$ can be equivalently rewritten as
\begin{align}
\|\{\sigma_i(\bm{M})\}_{i = 1}^M\|_0 = 1,
\end{align}
where $\sigma_i(\bm{M})$ denotes the $i$-th largest singular value of matrix $\bm{M}$.
Furthermore, since the trace norm and the spectral norm are represented by
\begin{align}
\mathrm{tr}(\bm{M}) = \sum\limits_{i = 1}^M {{\sigma _i}(\bm{M})}\ \textrm{and}\ \|\bm{M}\|_2 = \sigma_1(\bm{M}),
\end{align}
respectively, we have \cite{yang2020federated}
\begin{align}\label{DCrep}
\textrm{rank}(\bm{M})=1 \Leftrightarrow \mathrm{tr}(\bm{M})-\|\bm{M}\|_2 = 0,
\end{align}
with $\mathrm{tr}(\bm{M}) > 0$.
Therefore, we can use a DC penalty to induce rank-one solutions. The corresponding DC formulation for problem $\mathscr{P}_{1, 1}$ is given by
\begin{align}\label{Prob11DC}
\mathscr{P}'_{1, 1}:
\underset{\bm{x} \in \mathbb{R}_+^K, \bm{M}}{\text{minimize}} \hspace{3mm} & \|\bm{x}\|_1 + \rho \left(\mathrm{tr}(\bm{M}) - \|\bm{M}\|_2\right) \notag \\
\text{subject to} \hspace{3mm} & \mathrm{tr}(\bm{M}) - \gamma \cdot \mathrm{tr}(\bm{M} \bm{H}_i) \le x_i, \, \forall \, i \in \mathcal{K}, \notag \\
& \mathrm{tr}(\bm{M}) \ge 1, \, \bm{M} \succeq \bm{0},
\end{align}
where $\rho > 0$ denotes the penalty parameter.
Hence, we are able to obtain a rank-one matrix when the DC penalty term is enforced to be zero.
Then, the feasible aggregation beamforming vector $\bm{m}$ of problem $\mathscr{P}_1$ can be recovered by utilizing Cholesky decomposition for $\bm{M}^\star = \bm{m} \bm{m}^\mathsf{H}$.
Similarly, we detect the feasibility of problem $\mathscr{P}_{1, 2}$ by minimizing the DC representation term that is given by
\begin{align}\label{Prob12DC}
\mathscr{P}'_{1, 2}:
\underset{\bm{V}}{\text{minimize}} \hspace{2mm} & \mathrm{tr}(\bm{V}) - \|\bm{V}\|_2 \notag \\
\text{subject to} \hspace{2mm} & \|\bm{m}\|_2 - \gamma \left(\mathrm{tr}(\bm{R}_i \bm{V}) + |c_i|^2\right) \le x_i, \, \forall \, i \in \mathcal{K}, \notag \\
& \bm{V}_{n, n} = 1, \, \forall \, n \in \{1, \dots, N + 1\}, \notag \\
& \bm{V} \succeq \bm{0}.
\end{align}
Once the objective value becomes zero, we can obtain an exact rank-one feasible solution and extract $\bar{\bm{v}} = [\bm{v}_0, t_0]^\mathsf{T}$ by utilizing Cholesky decomposition for $\bm{V}^\star = \bar{\bm{v}} \bar{\bm{v}}^\mathsf{H}$.
Then, by computing $\bm{v} = \bm{v}_0 / t_0$, the phase-shift matrix can be recovered according to $\bm{\Theta} = \mathrm{diag}(\bm{v})$.

Problems $\mathscr{P}_{2, 1}$ and $\mathscr{P}_{2, 2}$ in the second step can be reformulated in the similar DC formulation to guarantee the feasibility of the rank-one constraint, which are rewritten as
\begin{align}
{\mathscr{P}'_{2, 1}}:
\underset{\bm{M}}{\text{minimize}} \hspace{3mm} & \mathrm{tr}(\bm{M}) + \rho \left(\mathrm{tr}(\bm{M}) - \|\bm{M}\|_2\right) \notag \\
\text{subject to} \hspace{3mm} & \mathrm{tr}(\bm{M} \bm{H}_i) \ge 1, \, \forall \, i \in \mathcal{S}^{[k]}, \notag \\
& \bm{M} \succeq \bm{0},
\end{align}
and
\begin{align}
{\mathscr{P}'_{2, 2}}:
\underset{\bm{V}}{\text{minimize}} \hspace{3mm} & \mathrm{tr}(\bm{V}) - \|\bm{V}\|_2 \notag \\
\text{subject to} \hspace{3mm} & \mathrm{tr}(\bm{R}_i \bm{V}) + |c_i|^2 \ge 1, \, \forall \, i \in \mathcal{S}^{[k]}, \notag \\
& \bm{V}_{n, n} = 1, \, \forall \, n \in \{1, \dots, N + 1\}, \notag \\
& \bm{V} \succeq \bm{0}.
\end{align}

\subsection{DC Algorithm}

Although the DC programs are still nonconvex, their problem structures of minimizing the difference of two convex functions can be exploited to develop an efficient DC algorithm \cite{DCProg} by successively linearizing the concave part.

Specifically, in the first step, the objective functions of problem ${\mathscr{P}'_{1, 1}}$ and problem ${\mathscr{P}'_{1, 2}}$ can be denoted as $g_1 - h_1$ and $g_2 - h_2$, respectively, where
\begin{align}
& g_1 =\|\bm{x}\|_1 + \rho \cdot \mathrm{tr}(\bm{M}), \quad h_1 = \rho \cdot \|\bm{M}\|_2, \\
& g_2 = \mathrm{tr}(\bm{V}), \quad h_2 = \|\bm{V}\|_2.
\end{align}
For problem ${\mathscr{P}'_{1, 1}}$, by linearizing the concave term $-h_1$ in the objective function, the resulting subproblem at the $t$-th iteration is given by
\begin{align}\label{ProbDCAlg1}
\underset{\bm{x} \in \mathbb{R}_+^K, \bm{M}}{\text{minimize}} \hspace{3mm} & g_1 - \left<\partial_{\bm{M}^{[t - 1]}} h_1, \bm{M}\right> \notag \\
\text{subject to} \hspace{3mm} & \mathrm{tr}(\bm{M}) - \gamma \cdot \mathrm{tr}(\bm{M} \bm{H}_i) \le x_i, \, \forall \, i \in \mathcal{K}, \notag \\
& \mathrm{tr}(\bm{M}) \ge 1, \, \bm{M} \succeq \bm{0},
\end{align}
where $\left<\bm{X}, \bm{Y}\right> = \Re[\mathrm{tr}(\bm{X}^\mathsf{H} \bm{Y})]$ is the inner product of two matrices, and $\partial_{\bm{X}^{[t-1]}}h$ denotes the subgradient of function $h$ with respect to $\bm{X}$ obtained at iteration $t - 1$.
Besides, problem ${\mathscr{P}'_{1, 2}}$ can be solved by iteratively solving
\begin{align}\label{ProbDCAlg2}
\underset{\bm{V}}{\text{minimize}} \hspace{3mm} & g_2 - \left<\partial_{\bm{V}^{[t - 1]}} h_2, \bm{V}\right> \notag \\
\text{subject to} \hspace{3mm} & \|\bm{m}\|^2 - \gamma \left(\mathrm{tr}(\bm{R}_i \bm{V}) + |c_i|^2\right) \le x_i, \, \forall \, i \in \mathcal{K}, \notag \\
& \bm{V}_{n, n} = 1, \, \forall \, n \in \{1, \dots, N + 1\}, \, \bm{V} \succeq \bm{0}.
\end{align}

Likewise, in the second step, we can also transform the DC programs $\mathscr{P}'_{2, 1}$ and $\mathscr{P}'_{2, 2}$ into such a series of subproblems to apply the DC algorithm. In particular, we denote the objective functions of problem $\mathscr{P}'_{2, 1}$ and problem $\mathscr{P}'_{2, 2}$ as $g_3 - h_3$ and $g_4 - h_4$, respectively, where
\begin{align}
& g_3 = (1 + \rho) \cdot \mathrm{tr}(\bm{M}), \quad h_3 = \rho \cdot \|\bm{M}\|_2, \\
& g_4 = \mathrm{tr}(\bm{V}), \quad h_4 = \|\bm{V}\|_2.
\end{align}
The solution $\bm{M}^{[t]}$ for $\mathscr{P}'_{2, 1}$ is obtained by solving
\begin{align}\label{ProbDCAlg3}
	\underset{\bm{M}}{\text{minimize}} \hspace{3mm} & g_3 - \left<\partial_{\bm{M}^{[t - 1]}} h_3, \bm{M}\right> \notag \\
	\text{subject to} \hspace{3mm} & \mathrm{tr}(\bm{M} \bm{H_i}) \ge 1, \, \forall \, i \in \mathcal{S}^{[k]}, \notag \\
	& \bm{M} \succeq \bm{0}.
\end{align}
Besides, the solution $\bm{V}^{[t]}$ for $\mathscr{P}'_{2, 2}$ is obtained by solving
\begin{align}\label{ProbDCAlg4}
\underset{\bm{V}}{\text{minimize}} \hspace{3mm} & g_4 - \left<\partial_{\bm{V}^{[t - 1]}} h_4, \bm{V}\right> \notag \\
\text{subject to} \hspace{3mm} & \mathrm{tr}(\bm{R}_i \bm{V}) + |c_i|^2 \ge 1, \, \forall \, i \in \mathcal{S}^{[k]}, \notag \\
& \bm{V}_{n, n} = 1, \, \forall \, n \in \{1, \dots, N + 1\}, \notag \\
& \bm{V} \succeq \bm{0}.
\end{align}

Therefore, we have $\partial_{\bm{M}} h_1 = \partial_{\bm{M}} h_3 = \rho \cdot \partial \|\bm{M}\|_2$ and $\partial_{\bm{V}} h_2 = \partial_{\bm{V}} h_4 = \partial \|\bm{V}\|_2$, where $\partial \|\bm{X}\|_2$ can be efficiently computed by $\bm{u}_1 \bm{u}_1^\mathsf{H}$ \cite{yang2020federated} and $\bm{u}_1$ is the eigenvector corresponding to the largest eigenvalue of the matrix $\bm{X}$.
Consequently, it can be verified that the above subproblems are convex and thus can be efficiently solved by using CVX \cite{cvx}.
Furthermore, it has been shown in \cite{DCProg} that the solving procedure with the DC algorithm always converges to the critical points of the DC programs from any feasible initial points. 

	

\subsection{Computation Complexity Analysis}

In our proposed Algorithm \ref{alg_two_step_alternating_dc}, we need to solve a sequence of SDP problems \eqref{ProbDCAlg1}, \eqref{ProbDCAlg2} in the first step, and \eqref{ProbDCAlg3}, \eqref{ProbDCAlg4} in the second step. 
To solve each SDP problem, the worst case computational complexity by using the second-order interior point method \cite{MatrixLift} is $\mathcal{O}((M^2 + K)^{3.5})$ in problems \eqref{ProbDCAlg1} and \eqref{ProbDCAlg3}, and is $\mathcal{O}((N^2 + K)^{3.5})$ in problems \eqref{ProbDCAlg2} and \eqref{ProbDCAlg4}.
Supposing those problems converging to critical points of the DC programs with $T > 1$ iterations, the computational cost of solving a DC program, i.e., one of problems $\mathscr{P}'_{1, 1}$, $\mathscr{P}'_{1, 2}$, $\mathscr{P}'_{2, 1}$, and $\mathscr{P}'_{2, 2}$, is $\mathcal{O}(T (M^2 + K)^{3.5})$ or $\mathcal{O}(T (M^2 + K)^{3.5})$.
Note that we merely need to solve the SDP problem once for problems $\mathscr{P}_{1, 1}$, $\mathscr{P}_{1, 2}$, $\mathscr{P}_{2, 1}$, and $\mathscr{P}_{2, 2}$ via SDR technique with simply dropping the rank-one constraints, i.e., $T = 1$ in this case.
The proposed DC algorithm has a higher computational complexity than the SDR method.
Nevertheless, the sacrifice of the computational complexity results in significant improvement on the system performance, which will be demonstrated in the following section.

\section{Simulation Results} \label{SecSimulation}

In this section, we present the simulation results to demonstrate the advantages of the proposed two-step alternating DC algorithm for FL with device selection.
The effectiveness of deploying an IRS for the AirComp-based FL system will also be illustrated.
We consider a three-dimensional coordinate system, where the antennas at the BS and the reflecting elements at the IRS are placed as a uniform linear array and a uniform rectangular array, respectively.
The locations of the BS and the IRS are, respectively, set as $(3,0,6)$ meters and $(0,100,6)$ meters, while the edge devices are distributed in the region of $([0, 6], [100, 106], 0)$ meters surrounding the IRS.
The path loss model is given by
\begin{align}
	L(d) = {C_0}{(d/{d_0})^{-\alpha}},
\end{align}
where $C_0$ denotes the path loss at the reference distance $d_0 = 1$ meter, $d$ is the link distance, and $\alpha$ is the path loss exponent.
All channels are assumed to suffer from Rician fading \cite{wu2019intelligent}, where the channel coefficient can be expressed as
\begin{align}
	\bm{\varrho} = \sqrt{\frac{\varsigma}{1 + \varsigma}} \bm{\varrho}_{\rm LoS} + \sqrt{\frac{1}{1 + \varsigma}} \bm{\varrho}_{\rm NLoS},
\end{align}
where $\varsigma$ is the Rician factor, $\bm{\varrho}_{\rm LoS}$ denotes the line-of-sight (LoS) component, and $\bm{\varrho}_{\rm NLoS}$ denotes the non-line-of-sight (NLoS) component.
In our simulations, the channel coefficients are given by $\bm{G} = \sqrt{L(d_{\rm BI})} \bm{\varrho}_{\rm BI}$, $\bm{h}_i^\mathrm{r} = \sqrt{L(d_{{\rm ID}, i})} \bm{\varrho}_{{\rm ID}, i}$, and $\bm{h}_i^\mathrm{d} = \sqrt{L(d_{{\rm BD}, i})} \bm{\varrho}_{{\rm BD}, i}$,
where $d_{\rm BI}$, $d_{{\rm ID}, i}$ and $d_{{\rm BD}, i}$ denote the distance between BS and IRS, the distance between IRS and device $i$, and the distance between BS and device $i$, respectively.
As in \cite{wu2020discrete}, the Rician factors of $\bm{\varrho}_{\rm IB}$, $\bm{\varrho}_{{\rm DI}, i}$, and $\bm{\varrho}_{{\rm DB}, i}$ are set to be $3$ dB, $0$, and $0$, respectively, and the path loss exponents for the BS-device channel, the BS-IRS channel, and the IRS-device channel are set to be 3.6, 2.2, and 2.8, respectively.
Unless stated otherwise, other parameters are set as follows: $C_0 = -30$ dB, $P_0 = 20$ dBm, $\sigma^2 = -90$ dBm, $\epsilon = 10^{-3}$, $K = 20$, $M = 20$, and $N = 64$.

To ensure the effectiveness of our proposed two-step alternating DC algorithm for device selection, we first show the convergence behaviours of the sparse inducing step and the feasibility detection step in Fig.~\ref{fig_convergence_step1} and Fig.~\ref{fig_convergence_step2}, respectively. It is observed that the objective values of problems $\mathscr{P}_1$ and $\mathscr{P}_2$ are both able to converge to the stationary points by accurately finding rank-one solutions with DC programming. 

\begin{figure}[t]
	\centering
	\includegraphics[scale=0.55]{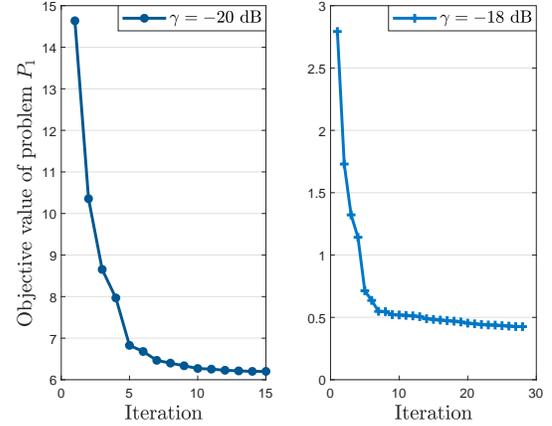}
	\caption{Convergence behaviour of sparse inducing step.}
	\label{fig_convergence_step1}
\end{figure}

\begin{figure}[t]
	\centering
	\includegraphics[scale=0.55]{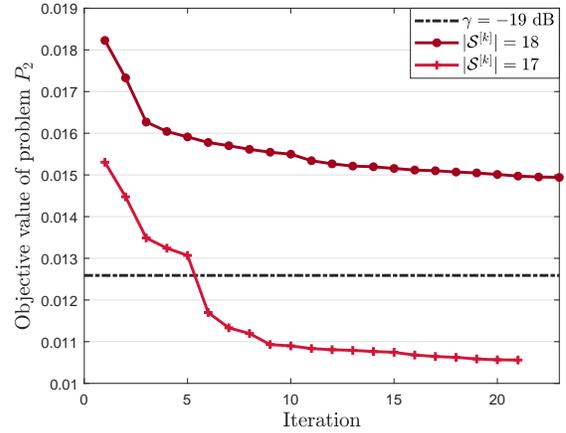}
	\caption{Convergence behaviour of feasibility detection step.}
	\label{fig_convergence_step2}
\end{figure}

Under the proposed two-step framework, we compare the proposed alternating DC based device selection algorithm (i.e., Algorithm \ref{alg_two_step_alternating_dc}) with the following baseline schemes:
\begin{itemize}
	\item \textbf{Alternating SDR with IRS}: In this scheme, the SDR method is applied to solve problems $\mathscr{P}_{1, 1}$, $\mathscr{P}_{1, 2}$, $\mathscr{P}_{2, 1}$, and $\mathscr{P}_{2, 2}$.
	
	\item \textbf{Random phase shifts}: In this scheme, the phase shift of each reflecting element at the IRS is uniformly and independently generated from $[0, 2 \pi)$. We merely solve problem $\mathscr{P}_{1, 1}$ in the first step and problem $\mathscr{P}_{2, 1}$ in the second step with the proposed DC algorithm.
	
	\item \textbf{Without IRS}: In the circumstance without IRS, only problem $\mathscr{P}_{1, 1}$ in the first step and problem $\mathscr{P}_{2, 1}$ in the second step need to be solved with the proposed DC algorithm by setting $\bm{\Theta} = \bm{0}$.
\end{itemize}

\subsection{Device Selection} \label{SecSim_DevSel}

Fig.~\ref{fig_mse} shows the average number of selected devices under different schemes versus the MSE threshold $\gamma$ for FL systems with and without IRS.
As the MSE threshold $\gamma$ increases, the average number of selected devices becomes larger. This is because reducing the requirement for the aggregation error is capable of inducing more edge devices to participate in the training process of FL.
In contrast to the scenario without IRS, deploying an IRS in the FL system can support much more devices for concurrent model aggregation under a certain MSE requirement. 
Besides, the scheme with random phase shifts performs worse than both the alternating DC and alternating SDR methods, which demonstrates the importance of jointly optimizing the device selection, the aggregation beamformer at the BS, and the phase shifts at the IRS.
Moreover, due to the effectiveness of obtaining the rank-one solutions with the DC algorithm, our proposed DC-based method significantly outperforms the SDR method.

\begin{figure}[t]
	\centering
	\includegraphics[scale=0.55]{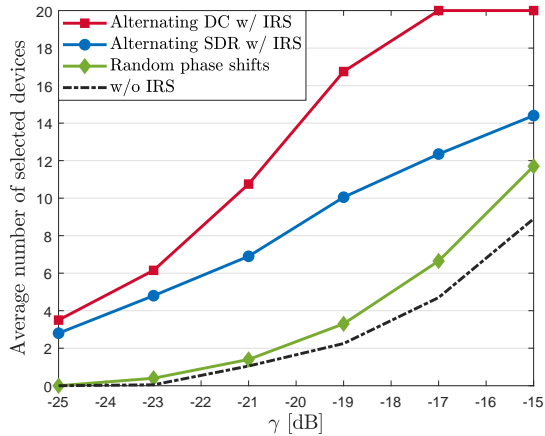}
	\caption{Average number of selected devices versus the MSE threshold.}
	\label{fig_mse}
\end{figure}

Fig.~\ref{fig_elements} illustrates the impact of the number of reflecting elements at the IRS on the average number of selected devices when $\gamma = -20$ dB.
As the number of reflecting elements increases, the IRS generates more accurate passive reflective beamforming for the incident signals, thereby effectively reducing the aggregation error at the BS. Therefore, the system is capable of selecting more edge devices to participate in FL, while satisfying the MSE requirement. In addition, since the SDR method has a high probability of failing to return rank-one solutions for high-dimensional optimization problems, it is observed that the gap between the DC and SDR schemes increases as the number of reflecting elements at the IRS becomes larger.

\begin{figure}[t]
	\centering
	\includegraphics[scale=0.55]{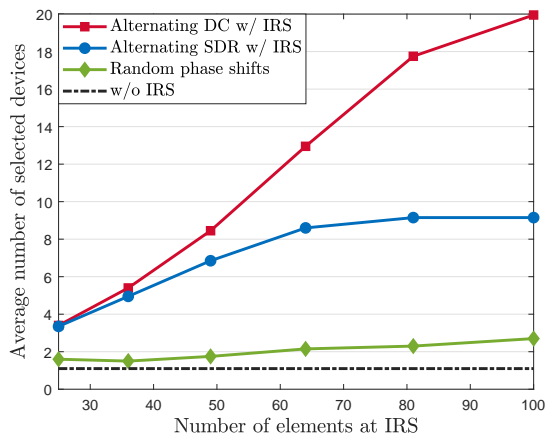}
	\caption{Average number of selected devices versus the number of reflecting elements at IRS.}
	\label{fig_elements}
\end{figure}

Fig.~\ref{fig_antennas} shows the impact of the number of antennas at the BS on the average number of selected devices when $\gamma = -22$ dB.
As the number of antennas at the BS increases, the channel gain between the BS and each edge device is enhanced by gathering signals from more antennas. Therefore, the adverse impact of additive noise at the BS can be alleviated and in turn the aggregation error is reduced, thereby being able to schedule more edge devices to participate in FL under a certain MSE requirement.
In addition, it is observed that even when the number of antennas at the BS is doubled, it is still difficult for the system without IRS to achieve a similar performance to the scenario with an IRS by jointly optimizing the aggregation beamformer at the BS and the phase shifts at the IRS. This observation implies that deploying an IRS not only enhances the system performance but also reduces the hardware complexity at the BS. Therefore, it is an efficient way to achieve fast and reliable model aggregation from the edge devices under a certain MSE requirement by deploying an IRS in AirComp-based FL system.

\begin{figure}[t]
	\centering
	\includegraphics[scale=0.55]{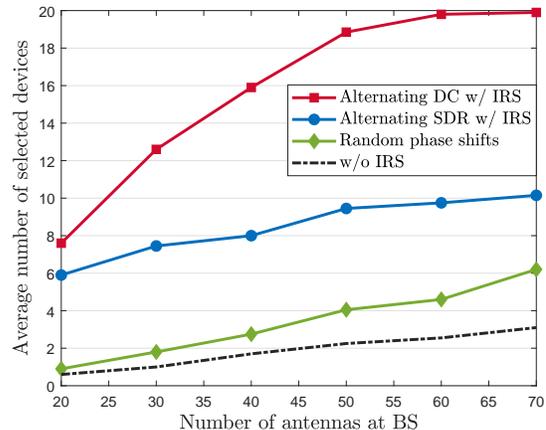}
	\caption{Average number of selected devices versus the number of antennas at BS.}
	\label{fig_antennas}
\end{figure}

\subsection{Performance Comparison for Federated Learning}

\begin{table*}[t]
	\centering
	\caption{Examples of handwritten digit identification with different schemes}
	\label{tab_model_test}
	\begin{tabular}{|c|c|c|c|c|c|c|c|}
		\hline
		\multirow{3}{*}{\textbf{Data}} & & & & & & & \\
		& \includegraphics[scale=0.14]{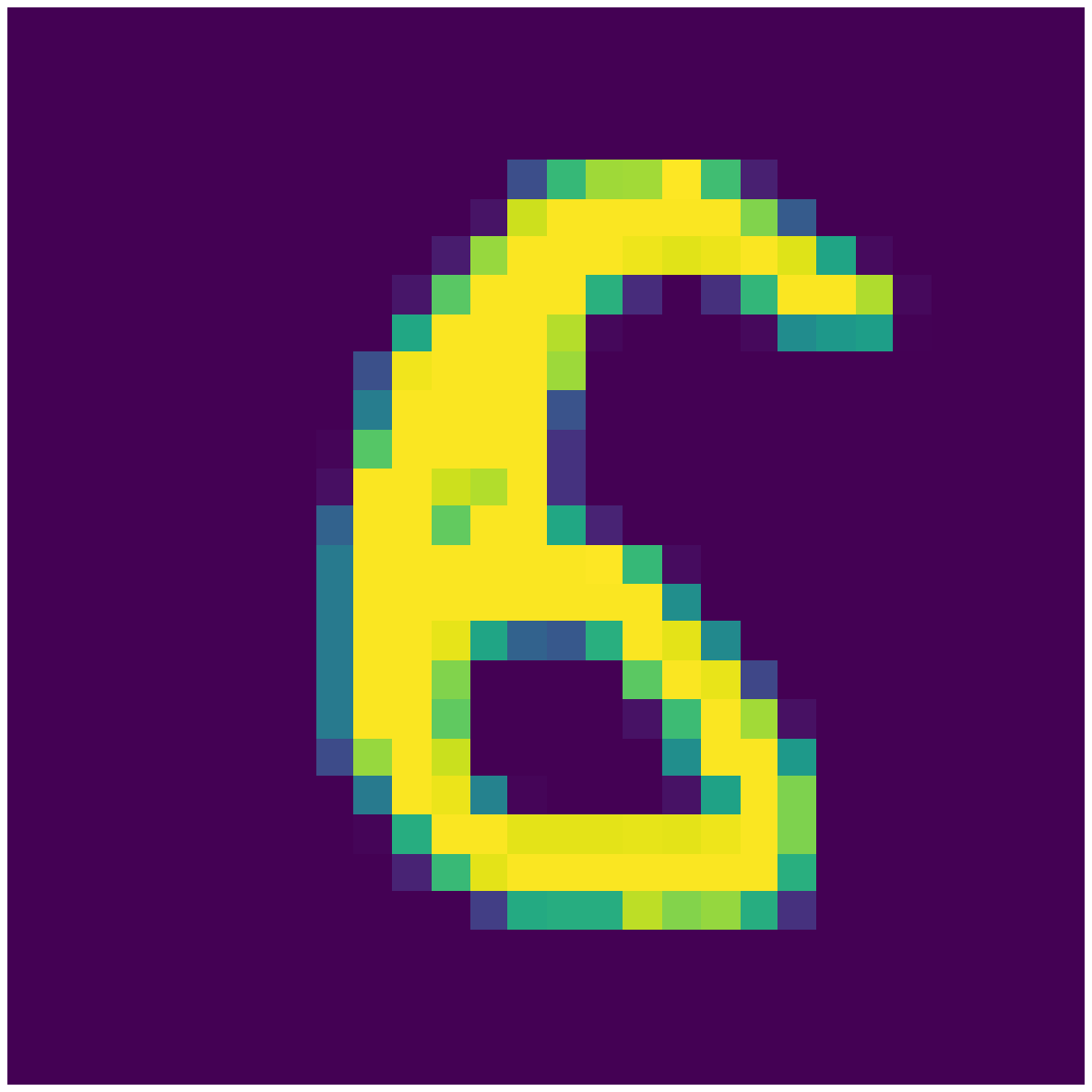}
		& \includegraphics[scale=0.14]{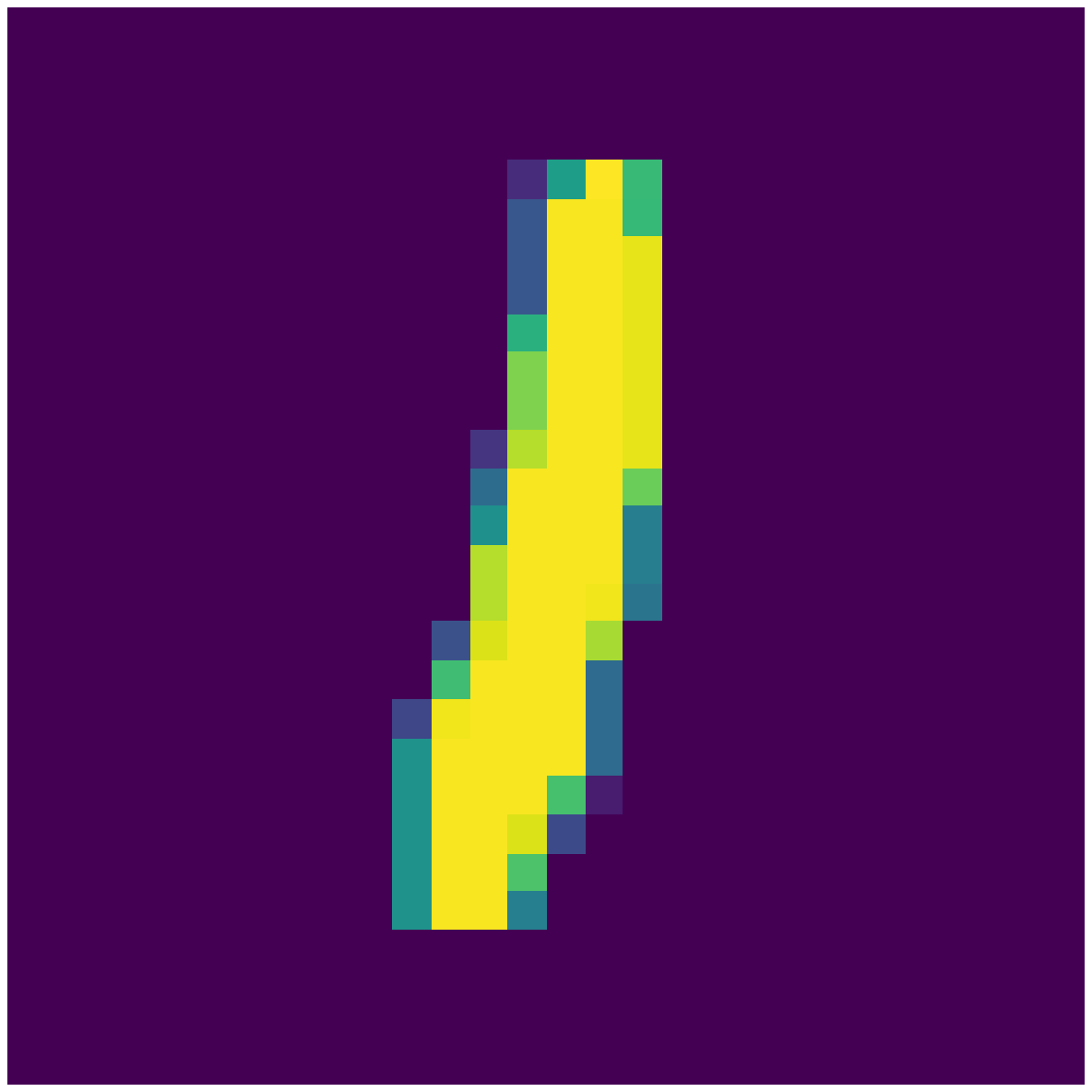}
		& \includegraphics[scale=0.14]{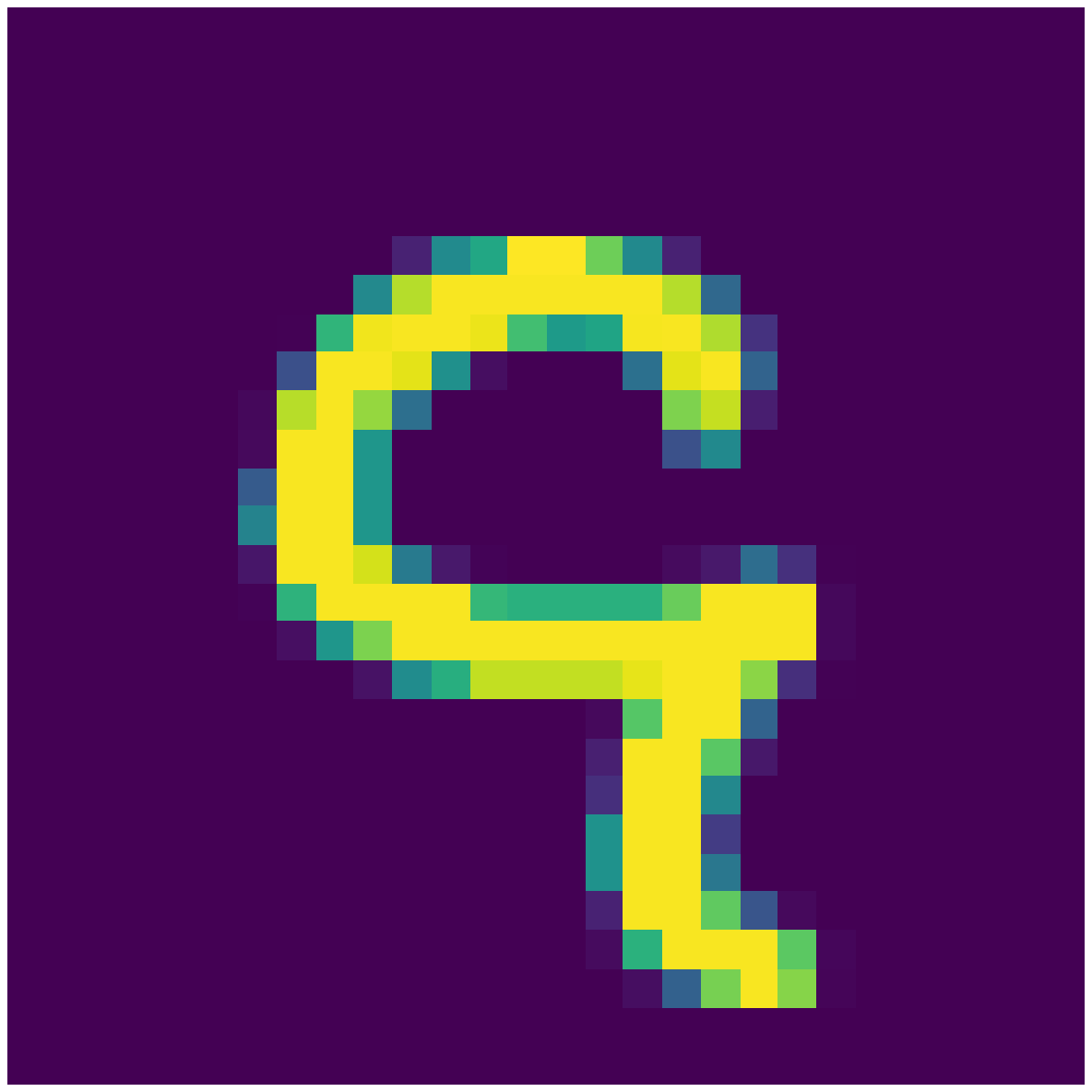}
		& \includegraphics[scale=0.14]{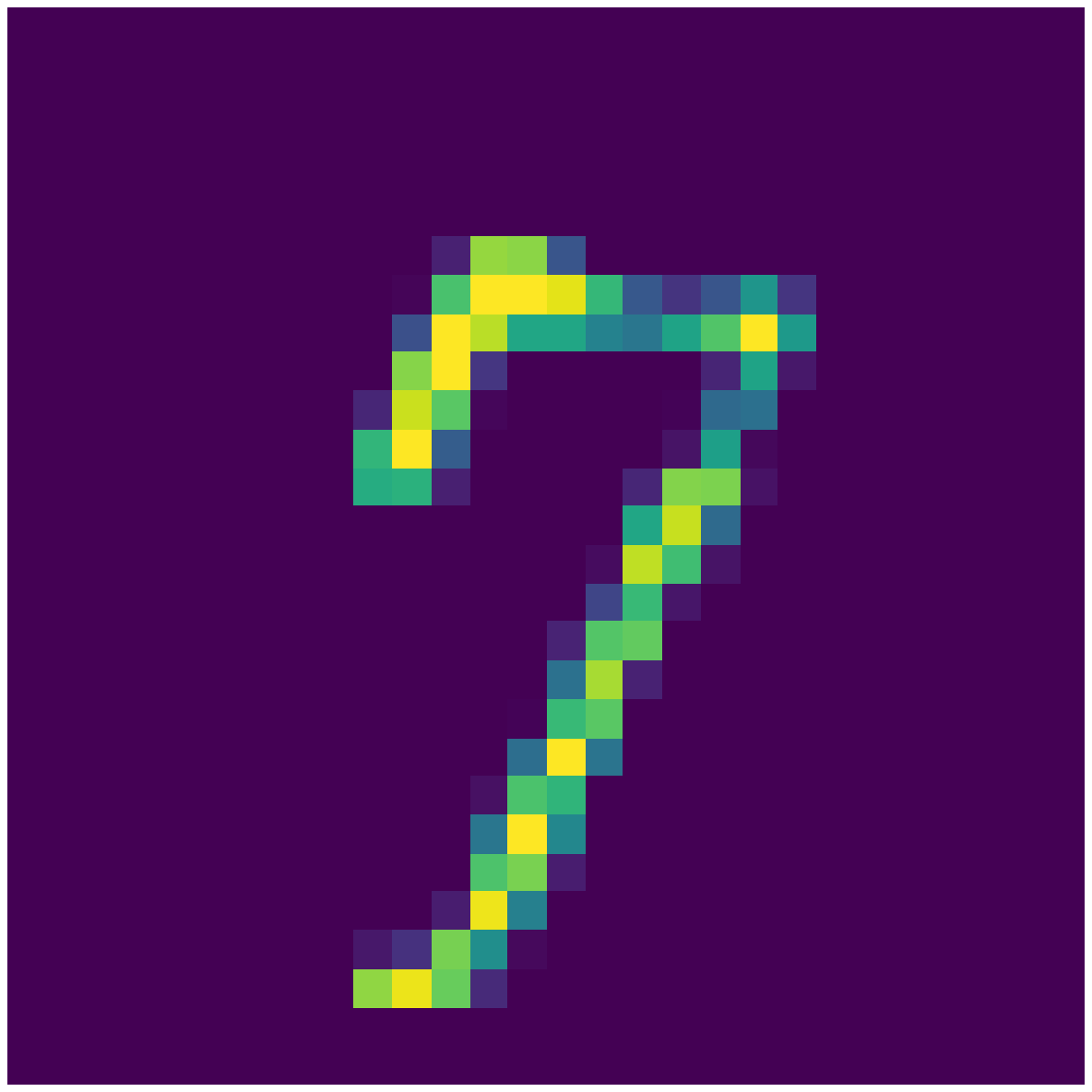}
		& \includegraphics[scale=0.14]{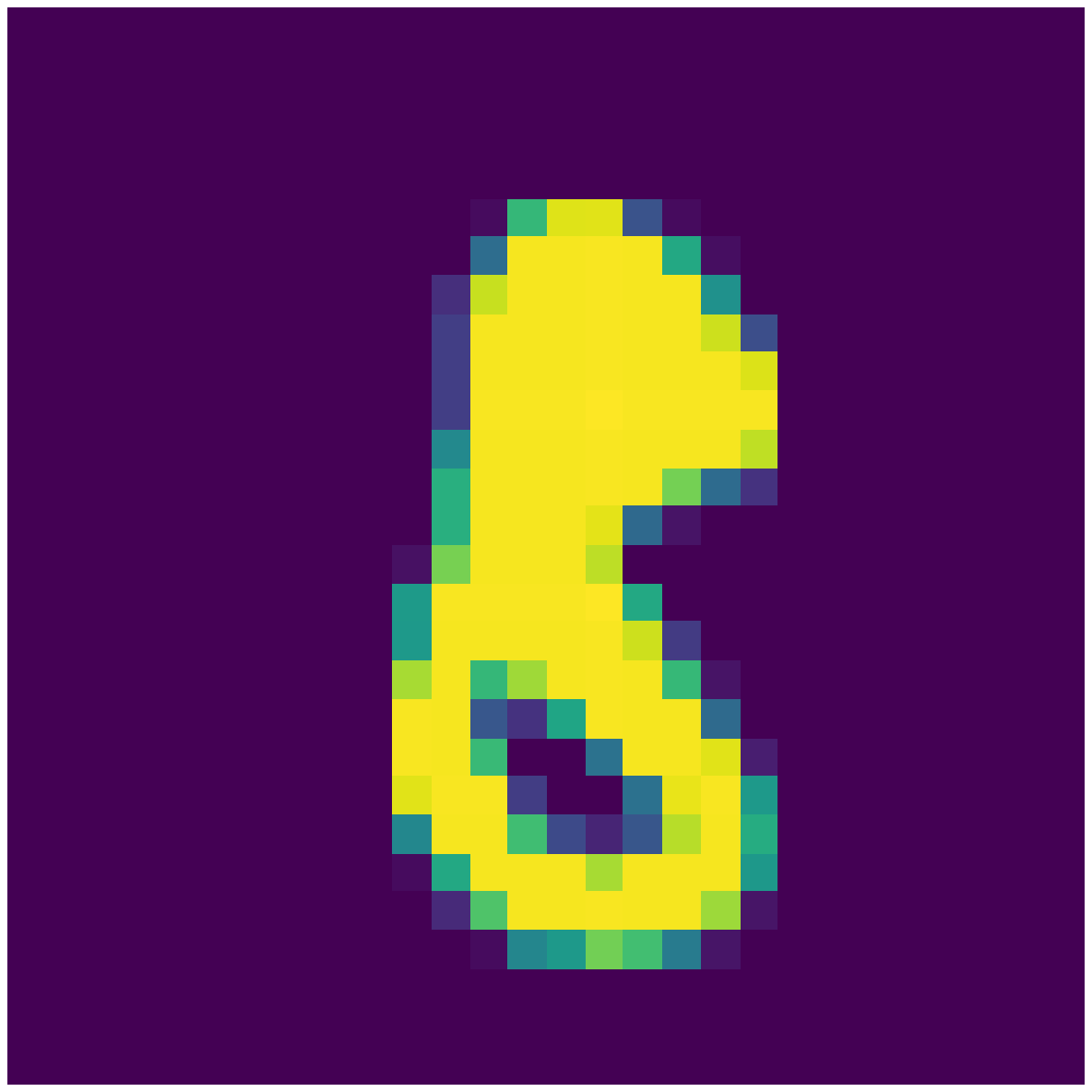}
		& \includegraphics[scale=0.14]{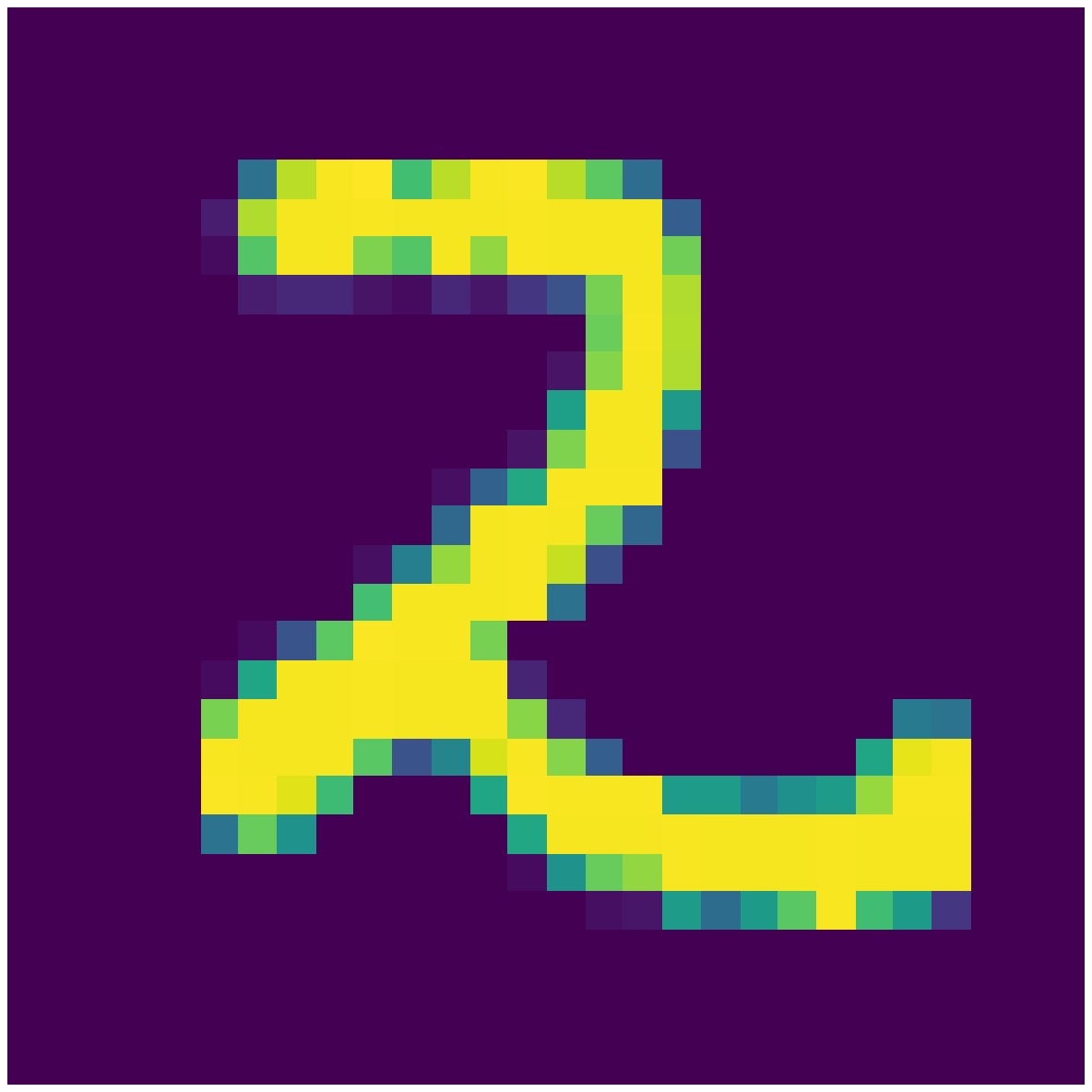}
		& \includegraphics[scale=0.14]{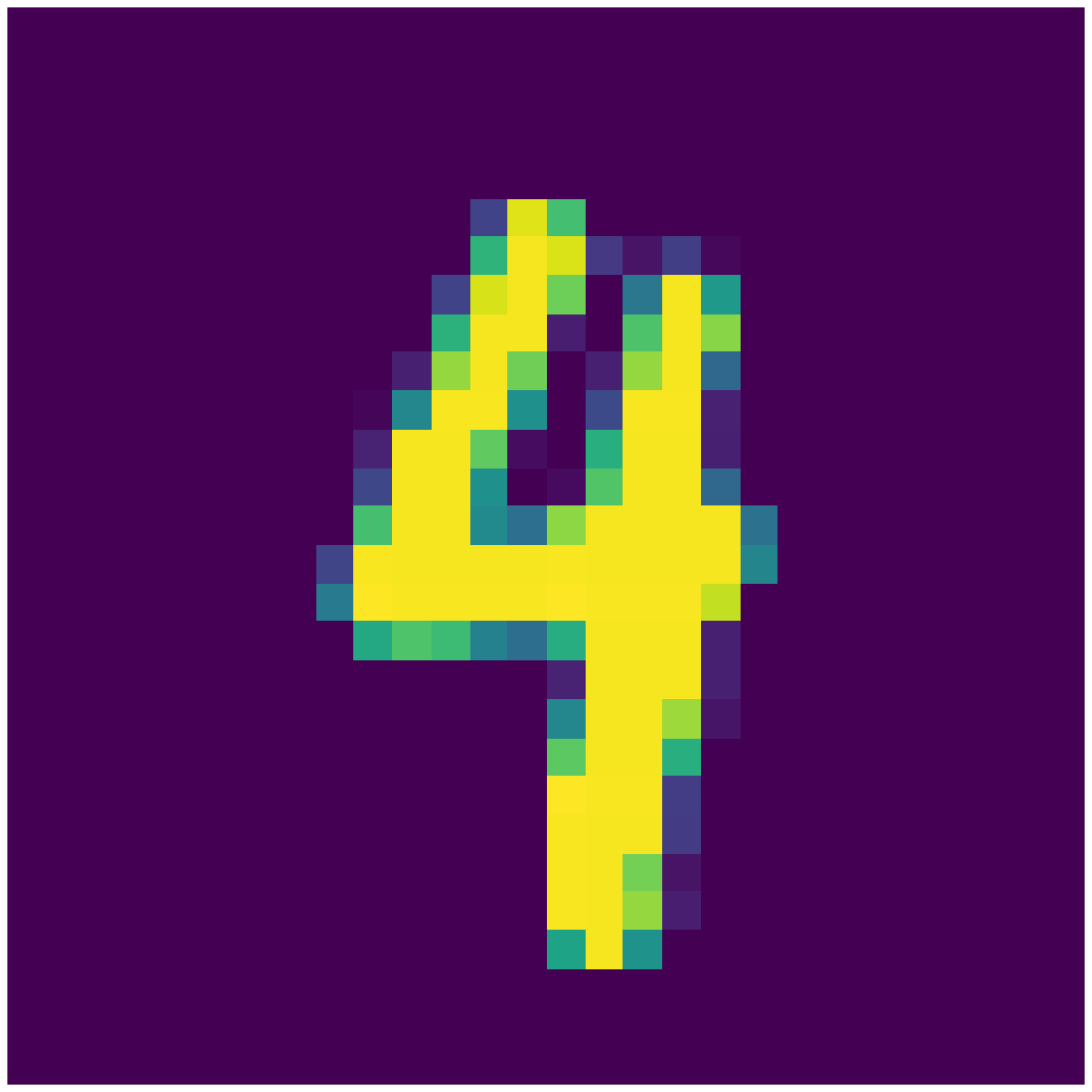} \\
		& & & & & & & \\
		\hline
		\multirow{1}{*}{\textbf{Label}} & 6 & 1 & 9 & 7 & 8 & 2 & 4 \\
		\hline
		\textbf{FL w/ IRS} & 6 & 1 & 9 & 7 & 8 & 2 & 4 \\
		\hline
		\textbf{FL w/o IRS} & {\color{red} 5 ($\times$)} & 1 & {\color{red} 4 ($\times$)} & {\color{red} 9 ($\times$)} & {\color{red} 5 ($\times$)} & 2 & 4 \\
		\hline
		\hline
		\multirow{3}{*}{\textbf{Data}} & & & & & & & \\
		& \includegraphics[scale=0.13]{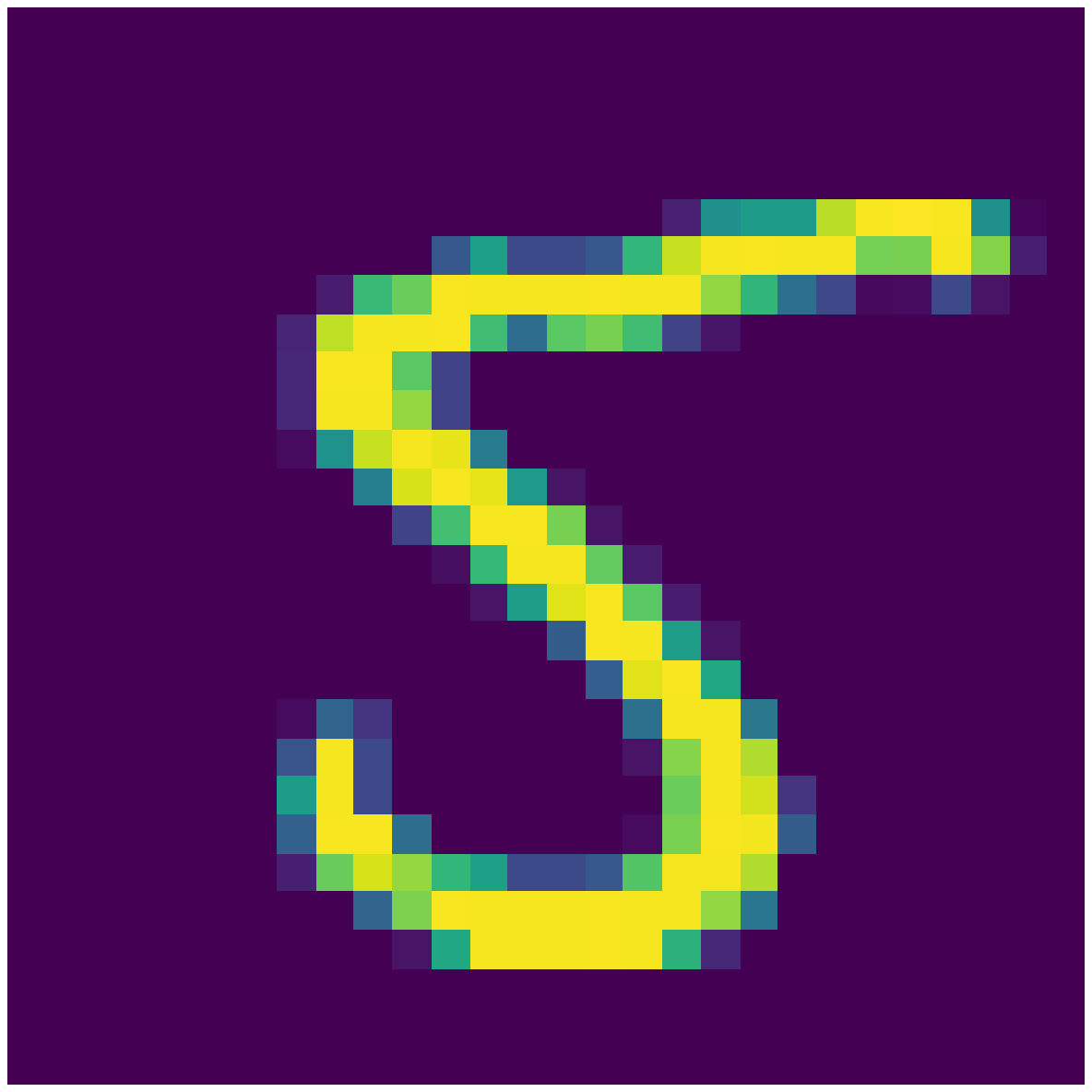}  
		& \includegraphics[scale=0.13]{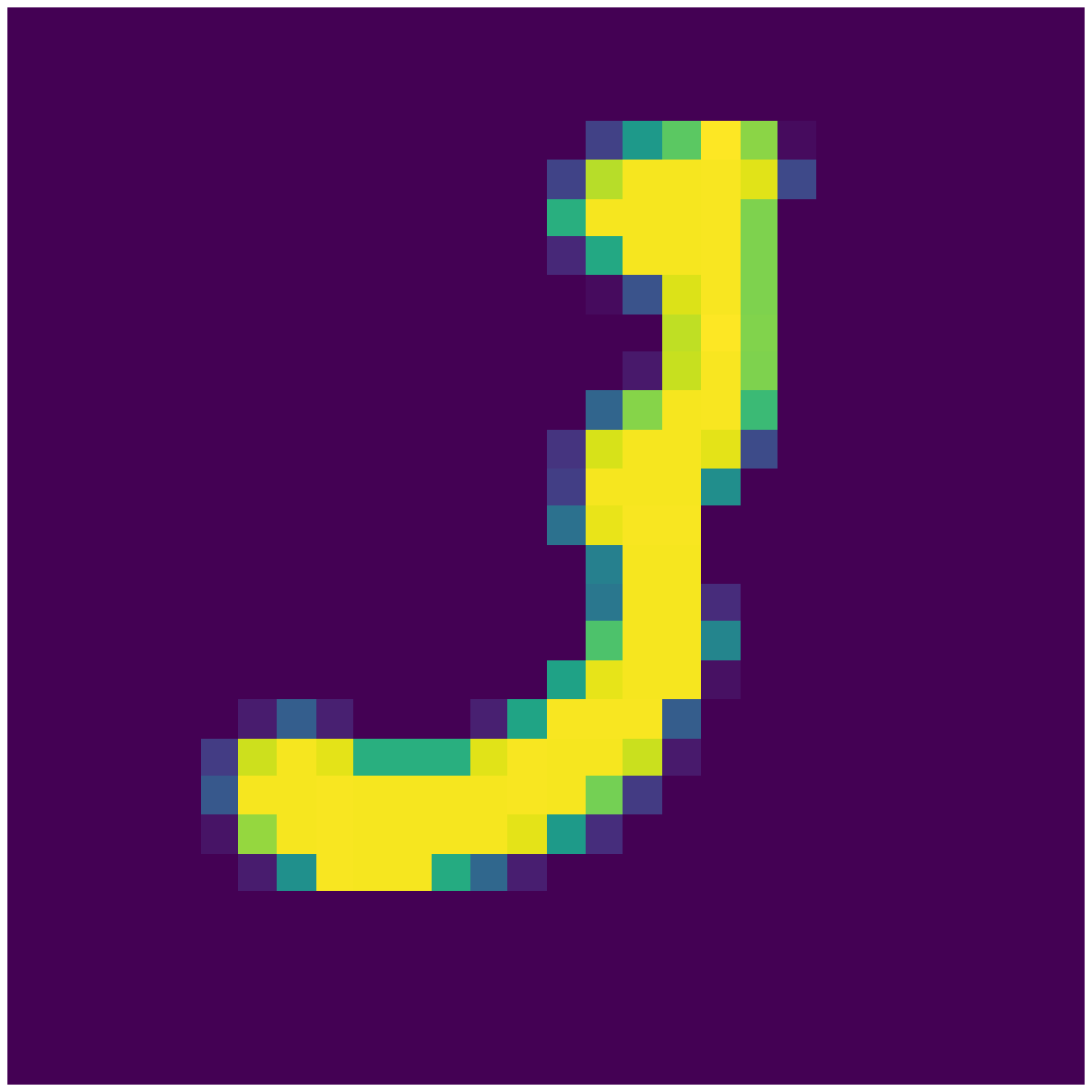}
		& \includegraphics[scale=0.13]{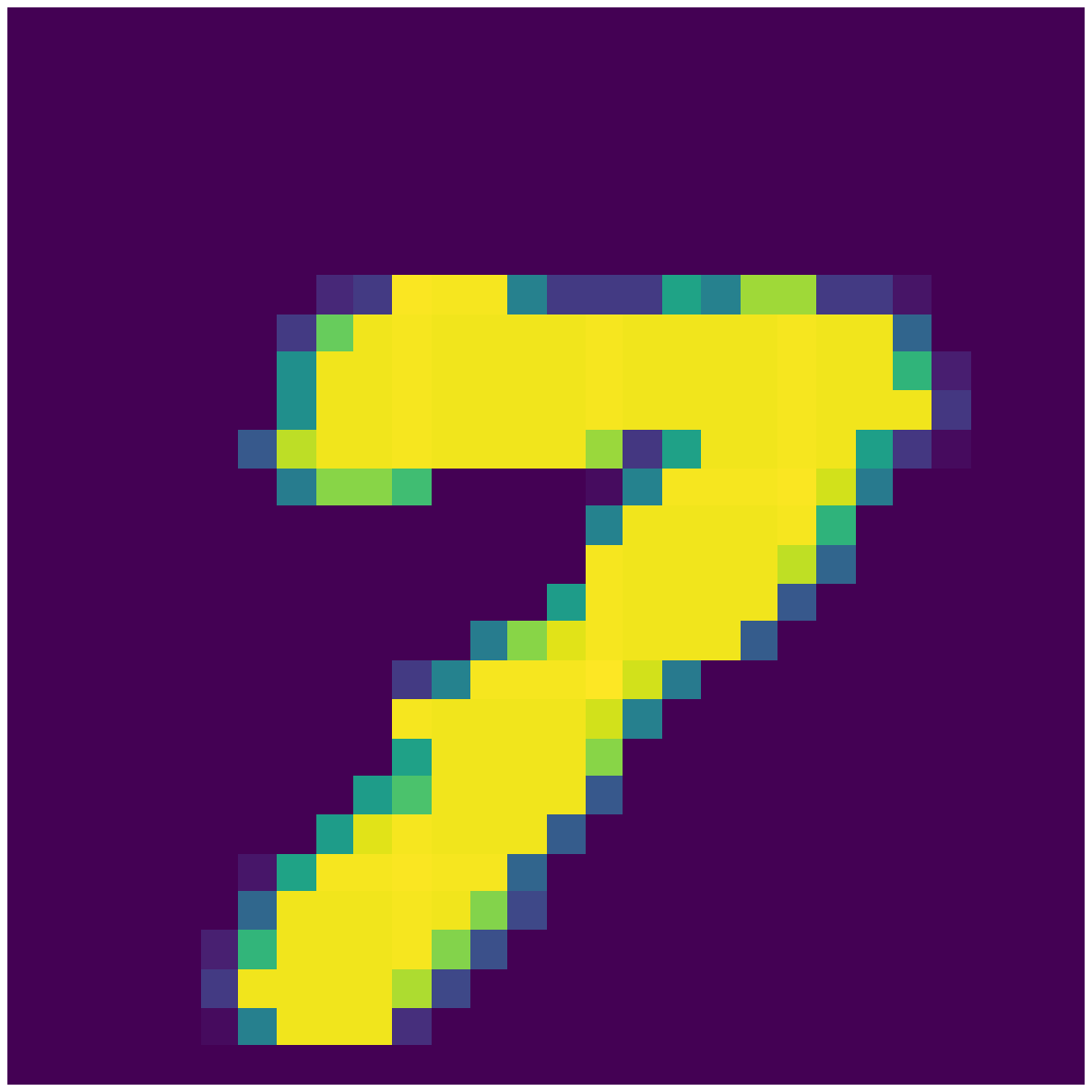}
		& \includegraphics[scale=0.13]{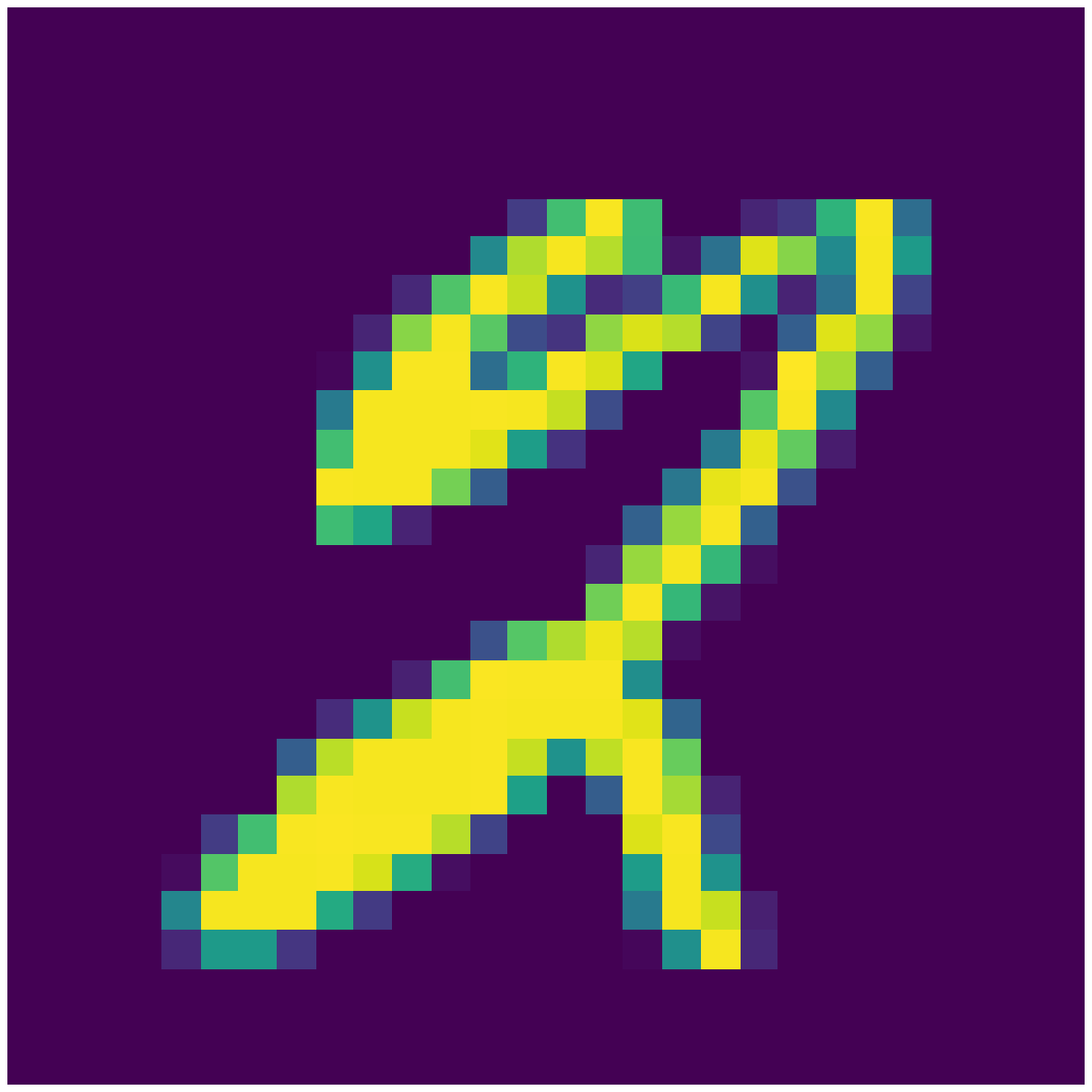}
		& \includegraphics[scale=0.13]{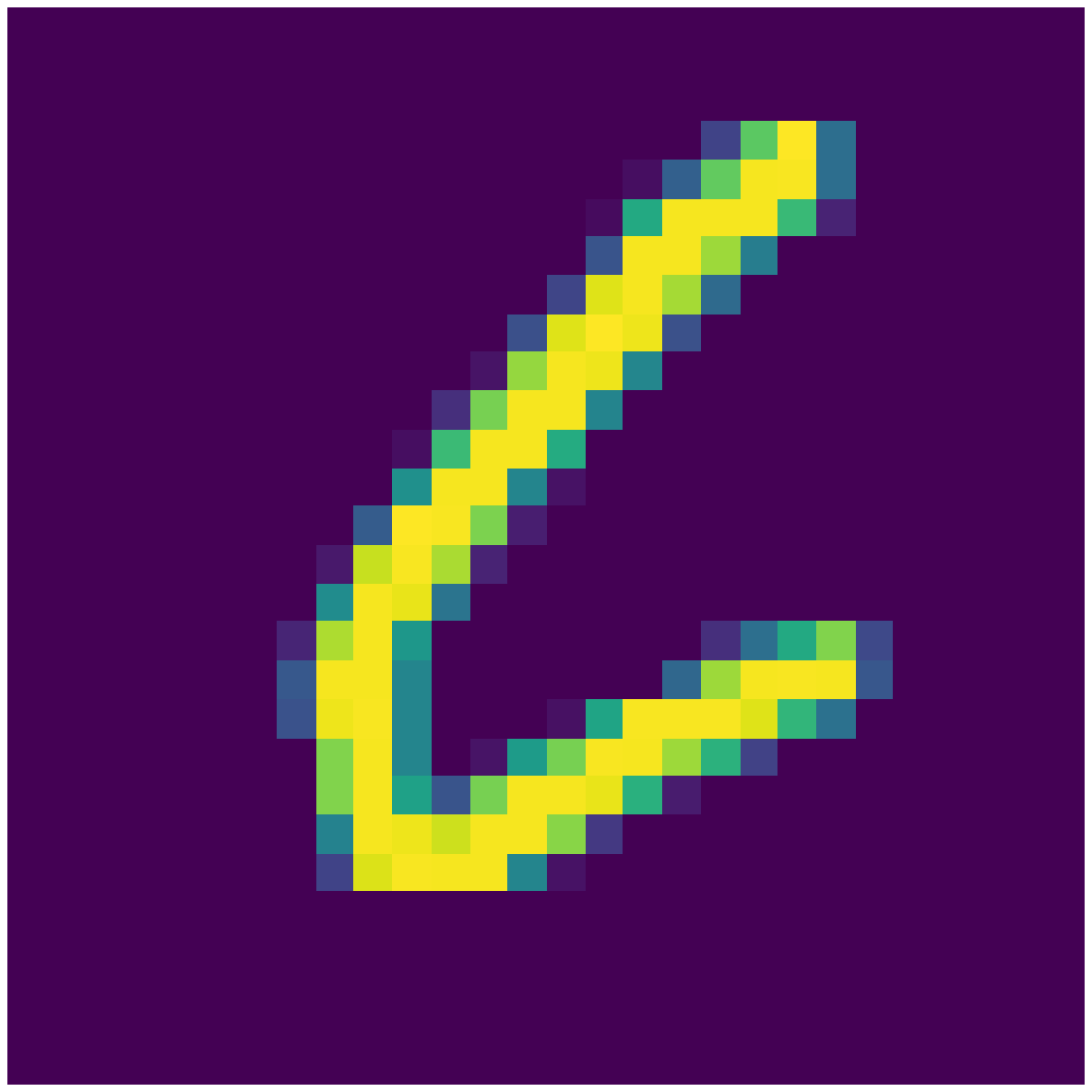}
		& \includegraphics[scale=0.13]{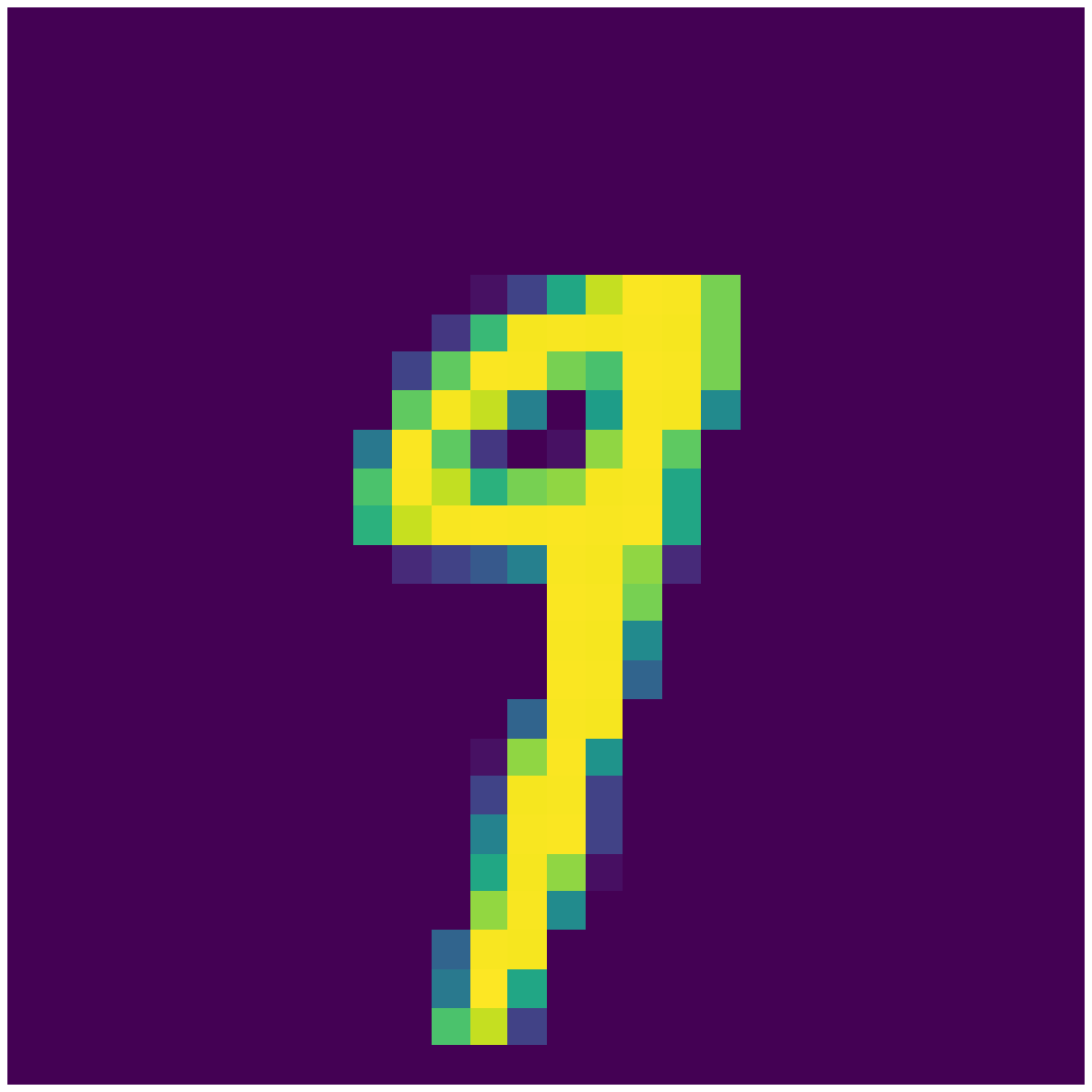}
		& \includegraphics[scale=0.14]{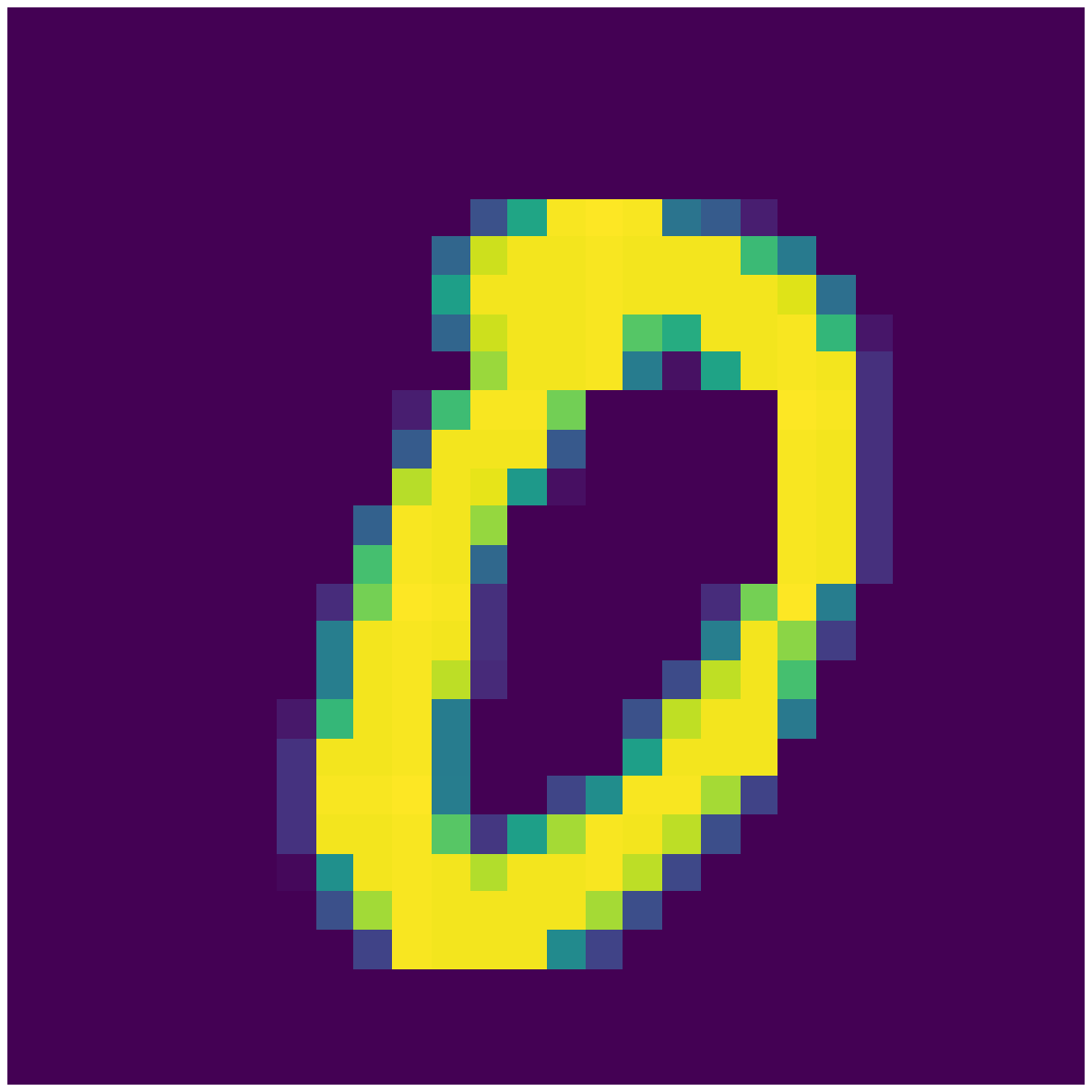} \\
		& & & & & & & \\
		\hline
		\multirow{1}{*}{\textbf{Label}} & 5 & 3 & 7 & 2 & 6 & 9 & 0 \\
		\hline
		\textbf{FL w/ IRS} & 5 & {\color{red} 2 ($\times$)} & 7 & 2 & 6 & {\color{red} 7 ($\times$)} & 0 \\
		\hline
		\textbf{FL w/o IRS} & 5 & {\color{red} 2 ($\times$)} & 7 & {\color{red} 4 ($\times$)} & {\color{red} 5 ($\times$)} & {\color{red} 4 ($\times$)} & 0 \\
		\hline
	\end{tabular}
\end{table*}

To directly show the excellent performance of the proposed two-step alternating DC algorithm for dealing with the FL tasks, we train a deep CNN on the widely used MNIST dataset.
In the simulations, the MNIST dataset consists of 10 classes with 6000 handwritten digits per class.
The case that all devices are selected at each communication round and without any aggregation error serves as the \textbf{Benchmark}.
We also consider that the dataset locally owned by each edge device is non-independent and identically distributed (non-IID) with uniform size.
In addition, each edge device only has two randomly selected classes.

When $\gamma = -17$ dB, the training loss and test accuracy averaged over 10 realizations are illustrated in Fig.~\ref{fig_training_loss} and Fig.~\ref{fig_test_accuracy}, respectively.
The results indicate that the proposed two-step alternating DC algorithm achieves a desirable performance in terms of a lower training loss and higher test accuracy than other schemes due to more edge devices and richer datasets being selected to participate in FL in each communication round.
In addition, the proposed algorithm achieves almost the same performance as the Benchmark scheme, which is an ideal case and serves as the performance upper bound.
Furthermore, Table \ref{tab_model_test} shows an example of handwritten digit identification under the FL system with and without IRS.
It can be observed that the IRS-assisted AirComp-based FL system is capable of achieving more accurate prediction with the given dataset, which indicates that under a specific MSE requirement, admitting more edge devices to participate in FL and collaboratively train a global model can achieve a lower training loss and higher test accuracy in fewer communication rounds.

\begin{figure}[t]
	\centering
	\includegraphics[scale=0.55]{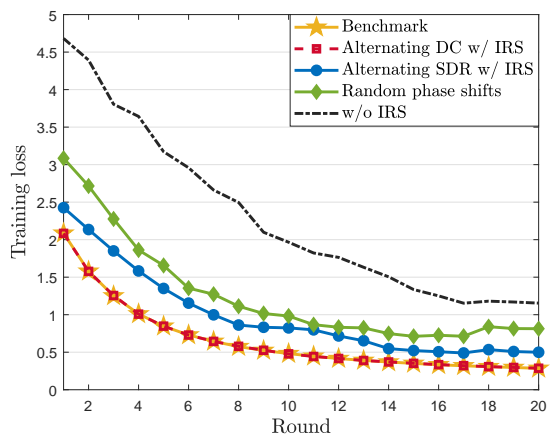}
	\caption{Training loss versus communication round.}
	\label{fig_training_loss}
\end{figure}

\begin{figure}[t]
	\centering
	\includegraphics[scale=0.55]{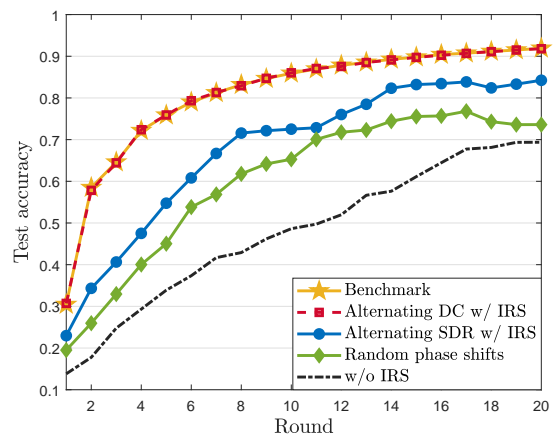}
	\caption{Test accuracy versus communication round.}
	\label{fig_test_accuracy}
\end{figure}


\section{Conclusions} \label{SecConclusion}

In this paper, we proposed a novel IRS-assisted AirComp approach for fast model aggregation in a FL system.
To accelerate the convergence and enhance the learning performance of FL, we developed a two-step alternating low-rank optimization framework to maximize the number of selected devices under the MSE requirement for model aggregation.
We presented a DC formulation for rank-one constrained problems in the alternating procedure, followed by proposing the DC algorithm for solving the resulting DC programs.
Simulation results demonstrated that our proposed algorithm can achieve a lower training loss and higher test accuracy by selecting more devices under certain MSE requirements compared with the baseline scheme without IRS.

\begin{appendices}

\section{Proof of Proposition \ref{transmitter}} \label{ap_transmitter}
The transmitter scalar $\{w_i\}$ in Eq.~\eqref{equMSE} has the zero-forcing structure to enforce
\begin{align}
\sum\limits_{i \in \mathcal{S}} {{\left| \frac{1}{\sqrt \eta }{\bm{m}^\mathsf{H}} (\bm{G\Theta h}_i^\mathrm{r} + \bm{h}_i^\mathrm{d}) {w_i} - 1\right|}^2} = 0.
\end{align}
Moreover, we have
$\mathsf{MSE}(\hat g,g) \geq \sigma^2 \left\|\bm{m} \right\|^2/\eta$ from Eq.~\eqref{equMSE}.
We thus obtain the form of zero-forcing transmitter scalar given in Proposition \ref{transmitter} which minimizes the MSE.

\section{Proof of Proposition \ref{QuadCons}} \label{ap_quadcons}
The constraint \eqref{consMSE} can be reformulated as $F_i(\bm{m}) = \|\bm{m}\|^2 - \gamma |\bm{m}^\mathsf{H} (\bm{G} \bm{\Theta} \bm{h}_i^\mathrm{r} + \bm{h}_i^\mathrm{d})|^2 \le 0, i \in \mathcal{S}$, where $\bm{m} \neq \bm{0}$.
We can further rewrite it as
$F_i(\bm{m} / \sqrt{\tau}) = F_i(\bm{m}) / \tau \le 0, i \in \mathcal{S}$, where $\|\bm{m}\|^2 \ge \tau$ and $\tau > 0$.
By introducing optimization variable $\tilde{\bm{m}} = \bm{m} / \sqrt{\tau}$, the constraint \eqref{consMSE} can be equivalently reformulated as
$\|\tilde{\bm{m}}\|^2 - \gamma |\tilde{\bm{m}}^\mathsf{H}(\bm{G} \bm{\Theta} \bm{h}_i^\mathrm{r} + \bm{h}_i^\mathrm{d})|^2 \le 0, i \in \mathcal{S}$,
where $\|\tilde{\bm{m}}\|^2 \ge 1$.
Thus, we obtain the equivalent form of constraint \eqref{consMSE} given in Eq.~\eqref{consquad}.

\end{appendices}

\bibliographystyle{IEEEtran}
\bibliography{BibIRS}

\end{document}